\documentclass[bj]{imsart}
\RequirePackage{amsthm,amsmath,amsfonts,amssymb}
\RequirePackage[numbers]{natbib}
\RequirePackage[colorlinks,citecolor=blue,urlcolor=blue]{hyperref}
\RequirePackage{graphicx}

\startlocaldefs
\theoremstyle{plain}

\newtheorem{theorem}{Theorem}[section]
\newtheorem{lemma}[theorem]{Lemma}
\newtheorem{corollary}[theorem]{Corollary}
\newtheorem{proposition}[theorem]{Proposition}

\theoremstyle{remark}
\newtheorem{definition}[theorem]{Definition}
\newtheorem*{example}{Example}

\newtheorem*{remark}{Remark}

\endlocaldefs

\usepackage[utf8]{inputenc}
\input{preamble.sty}

\newcommand{\revcolor}[1]{{\color{black}#1}}
\begin{document}

\begin{frontmatter}
\title{\revcolor{Combinatorial Bernoulli Factories}}
\runtitle{Combinatorial Bernoulli Factories}

\begin{aug}
\author[A]{\inits{R.}\fnms{Rad} \snm{Niazadeh}\ead[label=e1]{rad.niazadeh@chicagobooth.edu}},
\author[B]{\inits{R.}\fnms{Renato} \snm{Paes Leme}\ead[label=e2]{renatoppl@google.com}}
\and
\author[B]{\inits{J.}\fnms{Jon} \snm{ Schneider}\ead[label=e3]{jschneider@google.com}}
\address[A]{The University of Chicago Booth School of Business, Chicago, USA.
\printead{e1}}

\address[B]{Google Research, New York City, USA.\\
\printead{e2,e3}}
\end{aug}

\begin{abstract}
A Bernoulli factory is an algorithmic procedure for exact sampling of certain random variables having only Bernoulli access to their parameters. Bernoulli access to a parameter $p \in [0,1]$ means the algorithm does not know $p$, but has sample access to independent draws of a Bernoulli random variable with mean equal to $p$. In this paper, we study the problem of Bernoulli factories for polytopes: given Bernoulli access to a vector $x\in \mathcal{P}$ for a given polytope $\mathcal{P}\subset [0,1]^n$, output a randomized vertex such that the expected value of the $i$-th coordinate is \emph{exactly} equal to $x_i$. For example, for the special case of the perfect matching polytope, one is given Bernoulli access to the entries of a doubly stochastic matrix $[x_{ij}]$ and asked to sample a matching such that the probability of each edge $(i,j)$ be present in the matching is exactly equal to $x_{ij}$.

We show that a polytope $\P$ admits a Bernoulli factory if and and only if $\P$ is the intersection of $[0,1]^n$ with an affine subspace. Our construction is based on an algebraic formulation of the problem, involving identifying a family of Bernstein polynomials (one per vertex) that satisfy a certain algebraic identity on $\P$. The main technical tool behind our construction is a connection between these polynomials and the geometry of zonotope tilings.

We apply these results to construct an explicit factory for the perfect matching polytope. The resulting factory is deeply connected to the combinatorial enumeration of arborescences and may be of independent interest. For the $k$-uniform matroid polytope, we recover a sampling procedure known in statistics as Sampford sampling.\footnote{A preliminary conference version of this work has appeared in the proceeding of the 53rd Annual ACM SIGACT Symposium on Theory of Computing (STOC'21)~\citep{niazadeh2021combinatorial}. The current version presents all the missing proofs and technical details, as well as new results, alternative proofs, and more explanations.}

\end{abstract}


\end{frontmatter}
\section{Introduction}

Bernoulli factories are basic primitives used in statistics to generate exact samples of a random variable from independent samples of a related random variable. Bernoulli factory techniques have found their applications in settings as diverse as Bayesian mechanism design~\citep{dughmi2017bernoulli,cai2019efficient}, quantum physics~\citep{dale2015provable,yuan2016experimental}, exact simulation of stochastic processes such as diffusion \citep{blanchet2017exact}, Markov chain Monte Carlo (MCMC) methods~\citep{flegal2012exact}, and exact Bayesian inference~\citep{gonccalves2017exact,herbei2014estimating}. In mechanism design they allow for black-box reductions for welfare maximization that exactly preserve the Bayesian incentive compatibility, which offers stronger game-theoretical guarantees than approximately incentive compatible reductions. In Bayesian inference and stochastic simulation, the exact sampling afforded by Bernoulli factories allows them to be used in iterative methods without errors compounding.

In this paper we study Bernoulli factories for general polytopes -- with a particular focus on combinatorial settings. Before describing this (combinatorial) Bernoulli factory problem, it is useful to revisit the definition of the classic single-parameter version of the problem. The single parameter problem is typically phrased as generating new coins from old ones, where a coin here refers to a Bernoulli random variable. We are given access to a $p$-coin with unknown parameter $p$ and asked to generate a sample of an $f(p)$-coin for some known function $f:S \subseteq (0,1) \rightarrow (0,1)$. The algorithm does not know $p$, but has access to as many independent samples as it wants from a Bernoulli random variable with parameter $p$ (the $p$-coin); the goal is to output $1$ with probability $f(p)$.

For the function $f(p) = p^2$, for example, the algorithm can draw two samples from the $p$-coin and output $1$ if both samples are $1$, and outputs $0$ otherwise. A less trivial example is the function $f(p) = e^{p-1}$. Rewriting this function as the probability generating function of a discrete Poisson random variable, i.e., $f(p) = \E_{X \sim \textrm{Poisson(1)}}[p^X]$, leads to the following algorithm: (i) sample $X\sim \textrm{Poisson(1)}$, (ii) draw $X$ independent samples from the $p$-coin, and (iii) if all the samples are $1$ output $1$, otherwise output $0$. Keane and O'Brien~\citep{keane1994bernoulli} give necessary and sufficient conditions on function $f$ for the existence of Bernoulli factories.

Before we proceed, we emphasize a crucial point: the Bernoulli factory problem asks for \emph{exact} sampling, as opposed to (even very precise) approximate sampling. This property is essential the aforementioned applications in statistics, mechanism design and quantum mechanics, and is indeed the main motivation behind the study of Bernoulli factories. Approximate sampling is much simpler; in general, one can build an estimator $\hat{p}$ from i.i.d. samples and then sample a Bernoulli r.v. with parameter $f(\hat{p})$. This, however, is \emph{not} a Bernoulli factory.

\paragraph*{Combinatorial Bernoulli Factories} In this paper we will be mostly concerned with sampling a combinatorial object (e.g., a matching or a flow) having black-box sample access to marginal probabilities, say the probability that an edge is present in the matching. Formally, we are given an $n$-dimensional polytope $\P \subseteq [0,1]^n$ with vertices $V$. We are given $n$ coins with unknown probabilities $x_1,\hdots, x_n$ such that $x = (x_1, \hdots, x_n) \in \P \cap (0,1)^n$. A \emph{Bernoulli factory} for $\P$ is a then a randomized procedure for sampling a vertex $v \in V$ such that $\E[v]=x$. For convenience, the Bernoulli factory algorithm is allowed to use external randomness, besides using the given coins. This is indeed without loss of generality, since it is shown by~\citep{von} that it is possible to sample any random variable with known probability using a $p$-coin with unknown $p \in (0,1)$.



It is typical in the Bernoulli factory literature (e.g., \citep{keane1994bernoulli,nacu2005fast}) to restrict the input coins to be non-deterministic, i.e., $x_i \in (0,1)$. In some cases, though, it is possible to construct factories for all $x \in \P$ also allowing for $\{0,1\}$-coordinates (in other words, extending the factory to the boundary of $[0,1]^n$). We call a such a factory a \emph{strong Bernoulli factory}. We now ask the following question:
\begin{displayquote}
{\emph{\normalsize{Under what conditions does a polytope $\P\subseteq [0,1]^n$ admit a Bernoulli factory? 
If it admits one, how can one construct such a factory?}}}
\end{displayquote}
\noindent On the path to answer this question, it is useful to keep the following concrete examples in mind:


\begin{itemize}
    \item \bf{$\boldsymbol{k}$-subsets} (also known as \bf{$\boldsymbol{k}$-uniform matroids}): $n$ coins with unknown parameters $\{x_i\}_{i\in[n]}$ are given, such that $\sum_i x_i = k$ for some integer $k$. We are asked to sample a subset $S \subseteq [n]$ of $k$ elements such that $\Pr[i \in S] = x_i$. 
    
    This setting corresponds to the polytope $\P = \{x \in [0,1]^n| \sum_i x_i = k\}$, which essentially is the $k$-uniform matroid polytope. The vertices correspond to the indicator vectors of subsets $S$ of size $k$, i.e., bases of the $k$-uniform matroid. 
    
    \item \bf{Matchings}: Consider a complete bi-partite graph with an independent $x_{ij}$-coin for each edge such that the parameters incident to every node sum to $1$. We want to sample a perfect matching such that edge $(i,j)$ is included with probability $x_{ij}$.
    
    This setting corresponds to the Birkhoff-von Neumann polytope,
    $$\textstyle\P = \left\{x \in [0,1]^{n \times n}| \sum_k x_{kj} = \sum_k x_{ik} = 1, \forall i,j\right\},$$
    i.e., the set of doubly stochastic matrices. The vertices correspond to perfect matchings (or equivalently permutations over $[n]$) by the Birkhoff-von Neumann Theorem.
    
    \item \bf{Flows}: Consider a directed graph $(N,E)$ with a source $s$ and a sink $t$ and an $x_{ij}$-coin for each edge $(i,j) \in E$ such that for each node other than the source and the sink the sum of $x_{ij}$ for incoming edges is the same as the sum of $x_{ij}$ for outgoing edges. Let the sum of outgoing edges of the source $s$ is an integer $k$. We want to sample an integral $(s,t)$-flow of size $k$ such that edge $(i,j)$ is included with probability exactly $x_{ij}$.
    
    This setting corresponds to the flow polytope 
    $$\textstyle \P = \left\{x \in [0,1]^E| \sum_{j|(i,j) \in E} x_{ij} = \sum_{j| (j,i) \in E} x_{ji}, i \neq s,t; \sum_{j|(s,j) \in E} x_{sj} = \sum_{j|(j,t) \in E} x_{jt} = k \right\}.$$
    The vertices are integral $(s,t)$-flows. For $k=1$ this means sampling a path from $s$ to $t$.
\end{itemize}

\paragraph*{Main Result and Techniques} We answer the above question by providing necessary and sufficient conditions to construct combinatorial Bernoulli factories. More formally, we show it is necessary and sufficient that $\P$ is of the form $\H \cap [0,1]^n$, where $\H = \{ x \in \R^n| Wx = b \}$ is an affine subspace, for the existence of a Bernoulli factory for $\P$. The result is constructive and allows us to obtain factories for $k$-subsets, matchings, flows, and all other polytopes of the mentioned form.

The necessary condition is simpler and follows from an argument in polyhedral combinatorics. We show that if the polytope $\P$ is not of the form $\H \cap [0,1]^n$, there must exist two nearby points $x_1$ and $x_2$ in $\P$ and a vertex $v$ such that $x_2$ must output $v$ with non-negative probability while $x_1$ must output $v$ with zero probability (see Figure \ref{fig:2d_fac}). However, no algorithm can perfectly distinguish between $x$ and $x'$ with finitely many samples, so this is impossible.

The technically challenging part of this proof is to construct a factory for polytopes $\P$ of the form $\H \cap [0,1]^n$. Interestingly, we convert what is originally a probability problem to an algebraic question about Bernstein polynomials, which we then solve with the aid of techniques from geometric combinatorics. The main pieces of this argument are as follows:

\begin{itemize}
    \item \emph{Race over Bernstein polynomials}: We give a recipe for constructing Bernoulli factories by associating with each vertex $v$ of the polytope $\P$ a multivariate Bernstein polynomial $P_v(x)$ such that the polynomials satisfy $\sum_v P_v(x)(x-v) = 0$. This part of the proof follows from a combination of two ideas in the literature: (i) univariate Bernstein polynomials have been used many times to reason about Bernoulli factories in single-parameter settings, and (ii) the Bernoulli race construction of  \citep{dughmi2017bernoulli}.
    \item \emph{Generic and non-generic subspaces}: Each subspace $\H$ can be written in the form $Wx = b$ for an full-rank $k\times n$ matrix $W$.
    We say that a subspace is generic if the vertices of $\H \cap [0,1]^n$ have exactly $k$ coordinates in $(0,1)^n$. For any fixed $W$ the set of vectors $b$ for which $Wx = b$ is non-generic has measure zero.
    
    We first show how to construct strong Bernoulli factories for generic subspaces (see bullets below) and then obtain factories for non-generic subspaces as appropriately defined limits of generic factories. The non-generic construction is important since many polytopes of interest ($k$-subsets, matchings and flows) are non-generic.
    \item \emph{Polynomials from minors}: Given a generic affine subspace $\H$ of the form $Wx = b$ we can associate each vertex $v$ of the polytope $\P$ to a subset $S$ of size $k$ of variables that are \emph{basic} (in the terminology of the simplex method). The determinant of the subset of the $k \times k$ minor corresponding to the basic variables is then used to construct a Bernstein monomial associated with $v$.
    \item \emph{Zonotope tilings}: Finally, we need to show that this construction satisfies the polynomial identity $\sum_v P_v(x)  (x-v) = 0$ in the first bullet. This is done by associating each vertex of the polytope with a point in a geometric space. The polynomial identity is then proved by considering two distinct decompositions of this geometric space into zonotope tilings (see Figures \ref{fig:2d_pp}-\ref{fig:2d_pz}).
\end{itemize}

These ingredients lead to the following algorithm (\Cref{alg:bernoulli_generic}). In this algorithm, we use $W_S$ to denote the $k \times k$ submatrix formed by the columns of $W$ corresponding to indices in $S \subseteq [n]$ . We also assume that the description of the affine subspace is such that $\abs{\det W_S} \leq 1$ for each subset $S$ of size $k$ (this is without loss of generality, since one can always scale $W$ and $b$ to satisfy it).

\begin{algorithm}[H]
\caption{{\sc Bernoulli Factory for Generic Subspaces}}
\label{alg:bernoulli_generic}
\begin{algorithmic} 
\State Pick a vertex $v \in V$ uniformly at random.
\State Let $S = \{i \in [n]| v_i \in (0,1)\}$.
\State For each $i \notin S$, sample an $x_i$-coin. If the sample is not equal to $v_i$, restart.
\State For each $i \in S$ sample two $x_i$-coins. If the samples are $1$ and $0$, proceed. Otherwise, restart.
\State With probability $\abs{\det W_S}$ output $v$. With remaining probability, restart.
\end{algorithmic}
\end{algorithm}

\paragraph*{Consequence for $\boldsymbol{k}$-subset} For the $k$-subset problem, we recover the sampling procedure in statistics known as Sampford sampling~\citep{sampford1967sampling}. Our result shows that this particular procedure can indeed be implemented as a Bernoulli factory. Since the $k$-subset polytope is non-generic it is obtained via the limit of generic polytopes. While the factories we construct for generic polytopes are strong factories, the limit of these factories diverges on the boundary of $[0, 1]^n$ -- hence,  we obtain a Bernoulli factory in the limit but not a strong factory. Indeed this is also a feature of Sampford's original sampling process -- it also requires all probabilities to be strictly in $(0,1)$. 

One may ask whether it is possible to extend Sampford sampling to also allow for deterministic variables. We prove a impossibility result showing that can exist no strong Bernoulli factory for $k$-subset with exponential tails (i.e., the probability of requiring more than $t$ coins is at most $c^t$ for some $c$). In particular, this implies that it is impossible construct a strong Bernoulli factory for $k$-subset by running Bernoulli race over Bernstein polynomials.

\paragraph*{Consequence for matching} For the matching problem, our Bernoulli factory has a particularly nice combinatorial structure. This structure is somewhat surprising since it goes through combinatorial constructions that do not seem to be related to sampling perfect matchings at first glance. In particular, we show our Bernstein polynomials can be alternatively obtained by enumerating particular monomials, one for each rooted arborescence in the complete graph $K_n$. The final argument to show the desired polynomial identity relies on counting arborescences using variants of Kirchhoff's matrix-tree theorem~\citep{kirchhoff1847ueber,tutte1948dissection} and additional combinatorial arguments related to trees and cycle covers. 

Here is the sampling procedure for matchings (\Cref{alg:bernoulli_matching}). Recall that we have an $n\times n$ doubly-stochastic matrix $[x_{ij}]$ with all entries in $(0,1)$. For each $(i,j)$ we have access to independent samples of an $x_{ij}$-coin. Our goal is to sample a perfect matching such that each edge $(i,j)$ is included in the matching with probability exactly equal to $x_{ij}$. 

\begin{algorithm}[H]
\caption{{\sc Bernoulli Factory for Matching}}
\label{alg:bernoulli_matching}
\begin{algorithmic} 
\State Pick uniformly at random a permutation $\pi$ over $[n]$.
\State For each $i \in [n]$ sample the  $x_{i  \pi(i)}$-coin. If any sample is $0$, restart.
\State Pick uniformly at random a spanning tree of the complete graph $K_n$.
\State Let $T$ be the set of edges $(i,j)$ of the tree oriented toward vertex $1$.
\State For each edge $(i,j) \in T$ sample the $x_{i \pi(j)}$-coin. If any sample is $0$, restart.
%
\State Output the matching $\{(i, \pi(i))\}_{i \in [n]}$.
\end{algorithmic}
\end{algorithm}



\paragraph*{Paper Organization}
In Section \ref{sec:prelims} we provide a formal definition of a Bernoulli factory as a decision tree and describe how it can be constructed via Bernstein polynomials. In Section \ref{sec:matching_factory} we give a self-contained presentation of a factory for matching via a combinatorial construction. In Section \ref{sec:necessary} we give necessary conditions on $\P$ for the existence of Bernoulli factories. We show in the following two sections that those conditions are also sufficient. In Section \ref{sec:generic}, we construct strong factories for generic subspaces via the geometry of zonotope tilings. In Section \ref{sec:non_generic}, we describe how to obtain factories for non-generic subspaces as limits of factories for generic ones. Finally, in \Cref{sec:impossibility_race_over_bernstein} we give an impossibility result for constructing strong factories for the $k$-subset polytope with fast convergence.

\paragraph*{Further Related Work} Beyond the work of Keane and O'Brien~\citep{keane1994bernoulli}, several other papers have studied different constructions and fast Bernoulli factory algorithms for functions $f:(0,1)\rightarrow (0,1)$. \citep{nacu2005fast} give the necessary and sufficient conditions for the existence of fast Bernoulli factories (see \Cref{sec:impossibility_race_over_bernstein} for an equivalent definition). An alternative algorithm for general analytic functions is proposed in \citep{latuszynski2011simulating}. A fast Bernoulli factory for rational functions is proposed in~\citep{mossel2005new}. More recently,~\citep{morina2019bernoulli} show how to construct a more practical Bernoulli factory for rational functions using coupling from the past~\citep{propp1996exact}. Both of these results extend to the ``dice enterprise problem", where the goal is exact simulation of a multivariate rational mapping $f:\Delta^{m}\rightarrow\Delta^m$ between probability simplices. Faster Bernoulli factories for linear functions are studied in~\citep{huber2016nearly,huber2017optimal}. \citep{mendo2019asymptotically} studies near-optimal Bernoulli factories for power series and~\citep{goyal2012simulating} study a particular class of Bernstein polynomials. Finally, extending Bernoulli factory algorithms to quantum settings is studied in~\citep{patel2019experimental}. 

In addition to those mentioned earlier, Bernoulli factory techniques have recently found other applications in different corners of computer science and statistics. They have been successfully applied to exact simulations of diffusions~\citep{blanchet2017exact}, designing exact simulation methods using MCMCs for Bayesian inference~\citep{herbei2014estimating,gonccalves2017barker}, designing particle filters~\citep{schmon2019bernoulli}, and designing blackbox reductions in Bayesian mechanism design~\citep{dughmi2017bernoulli,cai2019efficient}.

Indirectly related to us is the line of work on efficient approximate sampling from particular family of distributions (e.g., maximum entropy) over combinatorial polytopes (e.g., matching and matroid polytopes), satisfying a given vector of marginals. For example, see \citep{singh2014entropy,straszak2017real,anari2018log}. Our work diverges from this literature by the fact that a Bernoulli factory algorithm has only Bernoulli access to the marginal vector, and that it should satisfy the marginals exactly. Also, some aspects of the Bernoulli factory problem resemble the exact simulation of MCMCs in different contexts \citep{asmussen1992stationarity,jerrum1996markov,propp1998coupling}.

\section{Preliminaries}
\label{sec:prelims}
We start by formally defining a general Bernoulli factory that captures both the standard single-parameter Bernoulli factory and our generalization to the Bernoulli factory for polytopes. We then show how to use particular polynomials to construct those. 
\subsection{General Bernoulli factories}\label{sec:bernoulli_factories}


Below we define a general Bernoulli factory outputting elements in a set $V$ using a decision tree. We note that this definition is not tied to any particular function $f$.

\begin{definition}[Bernoulli factory]
A \emph{Bernoulli factory} $\mathcal{F}$ with output in $V$ is represented by a (possibly infinite) rooted binary tree $\mathcal{T}$. Each node in $\mathcal{T}$ has either $2$ children (in which case it is a \textit{non-leaf}) or $0$ children (in which case it is a \textit{leaf}). Each non-leaf $w$ is labeled with one of the $n$ random variables $\{x_1, x_2, \dots, x_{n}\}$ or with a constant $c \in (0,1)$. When executing the protocol, we either flip the $x_i$-coin in the label or a $c$-coin with known probability $c$. The edges from a non-leaf $w$ to its two children are labelled $0$ and $1$, corresponding to the output of the coin flip at $w$. Each leaf node $\ell$ is labelled with some $v \in V$, representing the output of our Bernoulli factory upon reaching this leaf node. 

To execute the factory $\mathcal{F}$, we start at the root node and repeatedly follow the following procedure. If we are at a non-leaf node $w$, we flip the coin given by $w$'s label, receive a result $r \in \{0, 1\}$, and follow the edge with label $r$ to one of $w$'s children. If we are at a leaf node $\ell$, we simply output the label of $\ell$.
\end{definition}

We let $\mathcal{F}(x) \in V \cup \{\emptyset\}$ denote the random variable corresponding to the output of the Bernoulli factory $\mathcal{F}$ on input $x \in [0, 1]^n$ (if $\mathcal{F}(x)$ does not terminate, we write $\mathcal{F}(x) = \emptyset$). We say a factory $\mathcal{F}$ \textit{terminates almost surely (a.s.)} on a domain $S \subseteq [0, 1]^n$ if $\Pr[\mathcal{F}(x) = \emptyset] = 0$ for all $x \in S$. Moreover, if the tree $\mathcal{T}$ corresponding to a Bernoulli factory $\mathcal{F}$ is finite, we say that $\mathcal{F}$ is a \emph{finite Bernoulli factory}.

\begin{definition}[One-bit Bernoulli factory]
We say that $\mathcal{F}$ outputting in $\{0,1\}$ is a \emph{one-bit Bernoulli factory} for the function $f: S \rightarrow [0, 1]$ on $S$ if (i) $\mathcal{F}$ terminates a.s. on $S$, and (ii) $\Pr[\mathcal{F}(x) = 1] = f(x), \forall x \in S$. 
\end{definition}





\subsection{Bernstein polynomials}

Fix a Bernoulli factory $\mathcal{F}$, and consider a leaf node $\ell$ of $\mathcal{F}$. Let $\Pr[\mathcal{F}(x) \rightarrow \ell]$ denote the probability that $\mathcal{F}(x)$ terminates at leaf $\ell$. By multiplying out the probabilities of each transition along the path from the root of $\mathcal{F}$ to $\ell$, we can write $\Pr[\mathcal{F}(x) \rightarrow \ell]$ in the form

\begin{equation}\label{eq:leaf}
\Pr[\mathcal{F}(x) \rightarrow \ell] = c\prod_{i=1}^{n}x_{i}^{a_i}(1-x_i)^{b_i}
\end{equation}

\noindent
for some $c \in [0, 1]$ and non-negative integers $a_i, b_i$ (for example, $b_i$ is the number of times variable $x_i$ appears on the path to $v$ where we take the edge labelled 0). The expression on the right-hand side of \Cref{eq:leaf} is known as a \textit{Bernstein monomial}. 

\begin{definition}[Bernstein polynomial]\label{def:bernstein}
A \emph{Bernstein monomial} in $n$ variables is a polynomial of the form:
$M(x) = \prod_{i=1}^n x_i^{a_i} (1-x_i)^{b_i}$
for $a_i, b_i \in \Z^{\geq 0}$. We will say that $a_i + b_i$ is the degree with respect to variable $i$ of this monomial and denote it $\deg_i(M)$.
A \emph{Bernstein polynomial} in $n$ variables is a positive combination of finitely many Bernstein monomials:
$P(x) = \sum_{j=1}^{k}c_{j}M_j(x)$
for Bernstein monomials $M_j(x)$ and coefficients $c_{j} \in \R^{+}$.
\end{definition}

Note that we can write $\Pr[\mathcal{F}(x) = v]$ as the sum of $\Pr[\mathcal{F}(x) \rightarrow \ell]$ over all leaves $\ell$ with label $v$. This means we can write $\Pr[\mathcal{F}(x) = v]$ as a weighted series of Bernstein monomials; in particular, if $\mathcal{F}(x)$ is a finite Bernoulli factory, then $\Pr[\mathcal{F}(x) = v]$ is a Bernstein polynomial in $x$. One fact that will prove particularly useful is a partial converse to this: given any Bernstein polynomial $P(x)$, it is always possible to construct a Bernoulli factory for a suitably normalized version of $P(x)$.

\begin{lemma}\label{lem:bernstein_factory}
Let $P(x) = \sum_{j=1}^{k}c_{j}M_j(x)$ be a Bernstein polynomial in $n$ variables, and let $C = \sum_{j=1}^{k} c_j$. Then there exists a finite one-bit Bernoulli factory for $P(x)/C$.
\end{lemma}
\begin{proof}
Consider the following Bernoulli factory $\mathcal{F}$. We first sample a monomial $M_j(x)$ with probability $c_j/C$ (by using external randomness). Now, if $M_j(x) = \prod_{i=1}^n x_i^{a_i} (1-x_i)^{b_i}$, flip each coin $i$ a total of $\deg_{i}(M)$ times. If for each coin $i$, the first $a_i$ flips returned $1$ and the next $b_i$ flips returned $0$, output $1$ for the overall factory $\mathcal{F}$. Otherwise, output $0$.

Conditioned on sampling monomial $M_j$, we return $1$ with probability $\prod_{i=1}^n x_i^{a_i} (1-x_i)^{b_i} = M_{j}(x)$. Since we sample monomial $M_j$ with probability $c_j/C$, the total probability $\mathcal{F}(x) = 1$ equals
$\Pr[\mathcal{F}(x) = 1] = \sum_{j}\frac{c_{j}}{C}M_{j}(x) = \frac{P(x)}{C},$
as desired.
\end{proof}

We now describe a method for constructing a factory outputting in $V$ from a collection of one-bit Bernoulli factories for each element $v \in V$. This method is known as a \textit{Bernoulli race} and it was introduced in \citep{dughmi2017bernoulli}. We summarize its properties in the following theorem. 

\begin{theorem}[Bernoulli race]\label{thm:bernoulli_race}
Fix a domain $S \subseteq [0, 1]^n$. For each $v \in V$, let $\mathcal{F}_{v}$ be a one-bit Bernoulli factory implementing a function $f_{v}: S \rightarrow [0, 1]$. If
$\sum_{v \in V} f_{v}(x) > 0, \forall x \in X$, then there exists a Bernoulli factory $\mathcal{G}$ that terminates a.s. on $S$ and outputs $v \in V$ with probability $f_{v}(x)/\sum_{v'}f_{v'}(x)$.
\end{theorem}
\begin{proof}
Consider the following procedure for $\mathcal{G}(x)$:

\begin{enumerate}[label=(\roman*)]
    \item Sample a $v$ uniformly at random from $V$.
    \item Run the factory $F_{v}(x)$. If the factory returns $1$, output $v$. Otherwise, return to step (i).
\end{enumerate}

We claim that this procedure terminates a.s. on $S$ and outputs vertex $v$ with probability $f_{v}(x)/\sum_{v'}f_{v'}(x)$. To see this, first note that each iteration of this procedure terminates with probability $(\sum_{v'}f_{v'}(x))/\abs{V} > 0$. Since there is a positive chance of terminating each round, and since each individual factory $F_{v}$ terminates a.s. on $S$, this procedure will terminate a.s. on $S$.

Now, note that we can write
\begin{equation}\label{eq:race1}
\Pr[\mathcal{G}(x) = v] = \frac{f_{v}(x)}{V} + \left(1 - \frac{\sum_{v'}f_{v'}(x)}{V}\right)\Pr[\mathcal{G}(x) = v],
\end{equation}
since there is a $\frac{f_{v}(x)}{V}$ chance we output $v$ in any given round, and a $1 - \frac{\sum_{w}f_{v}(x)}{V}$ chance we restart the procedure. Rearranging \Cref{eq:race1}, we have $\Pr[\mathcal{G}(x) = v] = f_{v}(x)/\sum_{v'}f_{v'}(x)$, as desired.
\end{proof}

In our applications, we will specifically want to take Bernoulli races over finite Bernoulli factories implementing Bernstein polynomials. 

\begin{corollary}[Bernoulli race over Bernstein polynomials]\label{cor:race_over_bernstein}
For each $v \in V$, let $P_{v}(x)$ be a Bernstein polynomial in $n$ variables. Fix a domain $S \subseteq [0, 1]^n$. If $\sum_{v} P_{v}(x) > 0$ for all $x \in S$, then there exists a Bernoulli factory which terminates a.s. on $S$ and outputs $v$ with probability
$P_{v}(x) / \sum_{v'} P_{v'}(x)$.
\end{corollary}
\begin{proof}
By Lemma \ref{lem:bernstein_factory}, for each $v \in V$, there exists a $C_{v} \geq 1$ such that for any $C \geq C_{v}$, there exists a finite Bernoulli factory for $P_{v}(x)/C$ over $S$. Choose $C = \max_{v}C_{v}$, and run a Bernoulli race over factories implementing $P_{v}(x) / C$. By Theorem \ref{thm:bernoulli_race}, such a race will output $v$ with probability $P_{v}(x) / \sum_{v'} P_{w}(v')$, as desired.
\end{proof}




\subsection{Combinatorial factories}


Finally, we return to the main focus of this paper. Recall that we wish to, given $x_i$-coins corresponding to the coordinates of a point $x$ within some polytope $\P \subseteq [0, 1]^n$, output a vertex $v$ of $\P$ so that $\E[v] = x$. 

\begin{definition}[Bernoulli factory for a polytope $\P$]\label{def:factory_for_P}
Let $\P \subseteq [0, 1]^n$ be a polytope contained in the unit hypercube, and let $\tilde{\P} = \P \cap (0, 1)^n$. Let $V$ denote the set of vertices of $\P$. A \emph{Bernoulli factory} for $\P$ is a factory $\mathcal{F}$ outputting in $V$ such that 
$$\E[\mathcal{F}(x)] = x, \forall x \in \tilde{\P}.$$

If the factory terminates almost surely for all $x \in \P$ (as opposed to just $x \in \tilde{\P})$, we say it is a \emph{strong Bernoulli factory} for $\P$. 
\end{definition}

Our main tool for constructing polytope factories will be to assign a Bernstein polynomial to each vertex and run a Bernoulli race over such polynomials (see Corollary \ref{cor:race_over_bernstein}). 

\begin{theorem}\label{thm:factory_construction}
In the setting of Definition \ref{def:factory_for_P}, if $P_v(x)$ is a non-zero Bernstein polynomial in $n$ variables for each $v \in V$ satisfying the following vector equality:
    \begin{equation}\label{eq:factory}
        \sum_{v\in V} P_{v}(x)(v - x) = 0, \forall x \in \P.
    \end{equation}
then running a Bernoulli race over the polynomials $P_{v}(x)$ (per Corollary \ref{cor:race_over_bernstein}) results in a Bernoulli factory for $\P$. Moreover, if 
\begin{equation}\label{eq:not_vanish}
\sum_{v \in V} P_v(x) > 0, \forall x \in \P
\end{equation}
it results in a strong Bernoulli factory for $\P$.
\end{theorem}

\begin{proof}
Since non-zero Bernstein polynomials are strictly positive on $(0,1)^n$ this automatically guarantees a.s. termination on $\tilde{P}$. To check that  $\E[\mathcal{F}(x)] = x$, it is sufficient to re-arrange \Cref{eq:factory} as follows: 
$\sum_{v \in V}P_{v}(x)v = \sum_{v \in V}P_{v}(x)x.$ Dividing by $\sum_{v \in V}P_{v}(x)$ we obtain exactly $\E[\mathcal{F}(x)] = x$.
\end{proof}
\section{A Factory for Matching}\label{sec:matching_factory}
In this section, we construct a Bernoulli factory for the Birkhoff-von Neumann perfect matching polytope using a race over Bernstein polynomials, as described in \Cref{cor:race_over_bernstein}. In later sections we will see how to systematically construct such factories for general polytopes; for now we will simply demonstrate the factory through its corresponding polynomials and prove that it works.

Throughout this section, let $\mathcal{B}_{n} \subseteq [0, 1]^{n\times n}$ denote the $n^{\textrm{th}}$ Birkhoff-von Neumann polytope. This polytope contains all doubly stochastic $n$-by-$n$ matrices. By the Birkhoff-von Neumann theorem,  $\mathcal{B}_{n}$ has $n!$ vertices, each corresponding to one of the $n$-by-$n$ permutation matrices (e.g., see \citep{schrijver2003combinatorial}). Each permutation $\pi$ can in turn be thought of as a perfect matching in the complete bipartite graph $K_{n, n}$. Let $S_n$ be the set of permutations of $[n]$. We identify $\pi \in S_n$ with the $n$-by-$n$ permutation matrix $[I_{i,j}]_{n\times n}\in\{0,1\}^{n\times n}$, where $I_{i,j}=\mathbb{I}\{j=\pi(i)\}$. We will abuse notation and use $\pi$ to denote both a permutation and its corresponding matrix. 
\paragraph*{Overview}  Recall from \Cref{thm:factory_construction} that we can specify a factory for the polytope $\mathcal{B}_{n}$ by specifying a non-zero Bernstein polynomial $P_{\pi}(x)$ in $n^2$ variables for each vertex $\pi\in S_n$ of $\mathcal{B}_n$, satisfying \Cref{eq:factory} for $x\in\mathcal{B}_n$. To specify these polynomials, we identify their monomials with certain directed graphs. An \textit{arborescence rooted at $r$} is a directed graph $T$ where (i) it has no directed cycles and (ii) for any vertex $v$ of $T$, there is exactly one directed path from $v$ to $r$ (in other words, it is a directed spanning tree where all edges are oriented towards $r$). Let $\mathcal{T}_{r}(n)$ be the set of arborescences rooted at $r$ with $n$ labelled vertices $1,\ldots,n$. Each element of $\mathcal{T}_{r}(n)$ is a collection of directed edges. Fix an arbitrary root $r\in[n]$. Then consider the following polynomials:

\begin{equation}\label{eq:matching_def}
\forall \pi\in S_n:~~~P_{\pi}(x) = \prod_{i=1}^{n}x_{i, \pi(i)}\sum_{T \in \mathcal{T}_{r}(n)}\prod_{(u, v) \in T} x_{u, \pi(v)}.
\end{equation}

\begin{figure}[h]
\centering
\begin{tikzpicture}[scale=.7,decoration={
    markings,
    mark=at position 0.5 with {\arrow{>}}}
    ]

\node[circle,fill,inner sep=1.5pt] at (0,0) {};
\node[circle,fill,inner sep=1.5pt] at (2,0) {};
\node[circle,fill,inner sep=1.5pt] at (1,1.732) {};
\draw[postaction={decorate},line width=1pt] (0,0) -- (1,1.732);
\draw[postaction={decorate},line width=1pt] (2,0) -- (1,1.732);
\node at (-.3,0) {$2$};
\node at (2.3,0) {$3$};
\node at (1,2.1) {$1$};

\begin{scope}[xshift=100]
\node[circle,fill,inner sep=1.5pt] at (0,0) {};
\node[circle,fill,inner sep=1.5pt] at (2,0) {};
\node[circle,fill,inner sep=1.5pt] at (1,1.732) {};
\draw[postaction={decorate},line width=1pt] (0,0) -- (2,0);
\draw[postaction={decorate},line width=1pt] (2,0) -- (1,1.732);
\node at (-.3,0) {$2$};
\node at (2.3,0) {$3$};
\node at (1,2.1) {$1$};
\end{scope}

\begin{scope}[xshift=200]
\node[circle,fill,inner sep=1.5pt] at (0,0) {};
\node[circle,fill,inner sep=1.5pt] at (2,0) {};
\node[circle,fill,inner sep=1.5pt] at (1,1.732) {};
\draw[postaction={decorate},line width=1pt] (2,0) -- (0,0);
\draw[postaction={decorate},line width=1pt] (0,0) -- (1,1.732);
\node at (-.3,0) {$2$};
\node at (2.3,0) {$3$};
\node at (1,2.1) {$1$};
\end{scope}

\end{tikzpicture}
\caption{Arborescences in $\T_1(3)$ corresponding to the monomials in $P_\eps$.}
\label{fig:arb_3}
\end{figure}
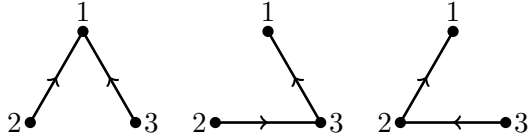

\begin{example}
\label{example:rooted-tree}
There are $3$ arborescences rooted at $1$ on $3$ vertices: $\{(2, 1), (3, 1)\}$, $\{(3, 1), (2, 3)\}$, and $\{(2, 1), (3, 2)\}$. For the identity permutation $\varepsilon$, we thus have (see \Cref{fig:arb_3}):

$$P_{\varepsilon}(x) = x_{1,1}x_{2,2}x_{3,3}\left(x_{2,1}x_{3,1} + x_{3,1}x_{2, 3}+ x_{2,1}x_{3,2}\right).$$

Replacing $\epsilon$ with a different permutation $\pi$ corresponds to applying $\pi$ to the second subscript of each variable. For example, for the permutation $\pi = (2, 3, 1)$, we have:

$$P_{\pi}(x) = x_{1,2}x_{2,3}x_{3,1}\left(x_{2,2}x_{3,2} + x_{3,2}x_{2, 1}+ x_{2,2}x_{3,3} \right).$$
 
\end{example}


The polynomials  $\{P_\pi(x)\}_{\pi\in S_n}$ defined above are clearly non-zero Bernstein polynomials, as required by \Cref{thm:factory_construction}. To show that $\mathcal{B}_n$ admits a Bernoulli factory, it only remains to show that the vector equality in \Cref{eq:factory} holds for all $x\in\mathcal{B}_n$. Note that in \Cref{eq:matching_def}, we choose an arbitrary vertex as the root of our arborescences in order to identify the polynomials $\{P_\pi(x)\}_{\pi\in S_n}$. Interestingly -- as we formally show in \Cref{prop:roots}-- the right hand side of \Cref{eq:matching_def} is the same for \emph{any} choice of root as long as we restrict our attention to points $x\in \mathcal{B}_n$. This property is indeed related to the fact that each polynomial $P_\pi(x)$ can be written as the product of a symmetric term and minor of a particular weighted directed graph Laplacian. More interestingly -- as we formally show in \Cref{prop:matching-factory} --  this property, together with a combinatorial argument relying on trees and permutations, are the keys in showing that \Cref{eq:factory} holds for points $x\in\mathcal{B}_n$. 
\begin{proposition}
\label{prop:roots}
Let $\mathcal{L}$ be the $n$-by-$n$ weighted directed Laplacian with (arc) weights $[x_{i,j}]_{n\times n}$, i.e., 
$$
\forall i,j\in[n]:~\mathcal{L}_{i,j} =
\left\{
	\begin{array}{ll}
		\sum_{k\neq i}x_{k,i}  & \mbox{if } i=j \\
		-x_{i,j} & \mbox{if } i\neq j
	\end{array}
\right.
$$
Moreover, for any $r\in[n]$, let $\mathcal{L}^{(r)}$ denote the $(n-1)$-by-$(n-1)$ submatrix of $\mathcal{L}$ obtained by removing the row and the column corresponding to $r$. Then for any $x \in \mathcal{B}_{n}$, and any $r,r' \in [n]$:
    \begin{equation} \label{eq:switchroot}
    \sum_{T \in \mathcal{T}_{r}(n)}\prod_{(u, v) \in T} x_{u, v} =\det[\mathcal{L}^{(r)}]= \det[\mathcal{L}^{(r')}]=\sum_{T \in \mathcal{T}_{r'}(n)}\prod_{(u, v) \in T} x_{u, v}.
\end{equation}
\end{proposition}
In order to prove the above proposition, we rely on two results from algebraic combinatorics. The first result is Tutte's matrix-tree theorem~\citep{tutte1948dissection}, which is essentially an adaptation of the standard Kirchoff's matrix-tree theorem~\citep{kirchhoff1847ueber} to weighted directed graphs. 

\begin{theorem}[\emph{Tutte's Matrix-Tree Theorem~\citep{tutte1948dissection}}]\label{thm:tutte}
Let $\mathcal{L}$ be the $n$-by-$n$ weighted directed Laplacian matrix with weights $[x_{i,j}]$. Then for any $r\in[n]$, 
$$
\det[\mathcal{L}^{(r)}]=\sum_{T \in \mathcal{T}_{r}(n)}\prod_{(u, v) \in T} x_{u, v}.
$$
\end{theorem}

The second result implies that all the principal minors of a \emph{zero-line-sum (ZLS)} matrix, i.e., a matrix whose rows and columns all sum to 0, are equal. Its proof can be found in \Cref{apx:matching} of the online supplement.
\begin{lemma}[\emph{ZLS matrices have equal cofactors}]\label{thm:zls}
Let $A$ be an $n$-by-$n$ matrix satisfying $\sum_{j=1}^{n} A_{ij} = 0$ for each $i \in [n]$ and satisfying $\sum_{i=1}^{n} A_{ij} = 0$ for each $j \in [n]$. For any $i,j\in[n]$, let $A^{(i,j)}$ denote the $(n-1)$-by-$(n-1)$ submatrix of $A$ obtained by removing the row and the column corresponding to $i$ and $j$ respectively. Then, for any $i,j,i',j' \in [n]$, 
$$(-1)^{(i+j)}\det[A^{(i,j)}] = (-1)^{(i'+j')}\det[A^{(i',j')}] ~~~~(\textrm{or equivalently}~~\cof_{i,j}[A]=\cof_{i',j'}[A]).$$
\end{lemma}

\begin{proof}[Proof of Proposition~\ref{prop:roots}]
First of all, \Cref{thm:tutte} implies that for any choice of $r\in[n]$,  $\det[\mathcal{L}^{(r)}]=\sum_{T \in \mathcal{T}_{r}(n)}\prod_{(u, v) \in T} x_{u, v}$. Next, note that every row in $\mathcal{L}$ sums to zero by definition. Moreover, every column in $\mathcal{L}$ also sums to zero when $x\in\mathcal{B}_n$, simply because $\mathcal{L}_{i,i}=1-x_{i,i}$ for $x\in\mathcal{B}_n$. Hence $\mathcal{L}$ is ZLS for $x\in\mathcal{B}_n$ and $\det[\mathcal{L}^{(r)}]=\det[\mathcal{L}^{(r')}]$ for any $r,r'\in[n]$ (due to \Cref{thm:zls}), as desired.
\end{proof}
Given \Cref{prop:roots}, we are now ready to prove the main result of this section, that is, our polynomials $\{P_\pi\}_{\pi\in S_n}$ satisfy \Cref{eq:factory} for all points in the perfect matching polytope.
\begin{theorem}
\label{prop:matching-factory}
 Consider the polynomials $\{P_\pi\}_{\pi\in S_n}$ as in \Cref{eq:matching_def}. Then for any $x \in \mathcal{B}_n$, 
\begin{equation}\label{eq:matching_constraint}
\sum_{\pi\in S_n}(\pi - x)P_{\pi}(x) = 0.
\end{equation}
\end{theorem}
\begin{remark}
Every step of the Bernoulli race procedure described in \Cref{thm:bernoulli_race}  -- over the Bernstein polynomials in \Cref{eq:matching_def}-- can be implemented efficiently by sampling uniform random permutations and uniform random spanning trees in the complete graph $K_n$. See \Cref{alg:bernoulli_matching}. 
\end{remark}

\subsection{Proof of Theorem~\ref{prop:matching-factory}}

We provide two different proofs of \Cref{prop:matching-factory}: a combinatorial proof (our original approach) and an algebraic proof (provided by Darij Grinberg). We believe both offer complementary insights and we therefore include both. The combinatorial proof is below and Darij's algebraic proof can be found in \Cref{appendix:algebraic_proof} of the online supplement.

Recall that \Cref{eq:matching_constraint} is an $n$-by-$n$ matrix equality. Fix any $(r,c) \in [n]\times[n]$; we will show that this equality holds for its $(r,c)^{\textrm{th}}$ entry. In other words, using the fact that  $\pi_{r,c}=\mathbb{I}\{\pi(r)=c\}$ in the permutation matrix $\pi$, we will show that
\begin{equation}\label{eq:matching2}
\sum_{\pi | \pi(r) = c} P_{\pi}(x) = \sum_{\pi}x_{r, c}P_{\pi}(x).
\end{equation}
Since $x \in \mathcal{B}_n$, the sum of each row and column of $x$ is equal to $1$. In particular, $\sum_{i=1}^{n}x_{r, i} = 1$. Therefore, by multiplying the LHS of \Cref{eq:matching2} and $\sum_{i=1}^n x_{r,i}$, it suffices to show that

\begin{equation}\label{eq:matching4}
\sum_{i=1}^{n}x_{r, i}\sum_{\pi | \pi(r) = c} P_{\pi}(x) = x_{r, c}\sum_{\pi}P_{\pi}(x).
\end{equation}
Recall from Proposition~\ref{prop:roots} that polynomials $\{P_\pi\}_{\pi\in S_n}$ are invariant to the choice of the root of arborescences used in \Cref{eq:matching_def}. Suppose $r$ is used as the root for all $\pi$ when defining $P_{\pi}(x)$. We will show that \Cref{eq:matching4} is true as a polynomial identity -- i.e. it is true not just for $x$ in $\mathcal{B}_n$, but for all $x \in \R^{n^2}$. To do so, it is enough to show that for any fixed $i \in [n]$,

\begin{equation}\label{eq:submatching}
x_{r, i}\sum_{\pi | \pi(r) = c} P_{\pi}(x) = x_{r, c}\sum_{\pi | \pi(r) = i}P_{\pi}(x).
\end{equation}
Summing \Cref{eq:submatching} over all $i\in[n]$, we obtain \Cref{eq:matching4}, as desired. Now, recall the definition of $P_\pi(x)$ when $r$ is used as the root:
\begin{equation}
\label{eq:rootatr}
P_{\pi}(x) = \prod_{i=1}^{n}x_{i, \pi(i)}\sum_{T \in \mathcal{T}_{r}(n)}\prod_{(u, v) \in T} x_{u, \pi(v)}.
\end{equation}
Note that in an arborescence rooted at $r$, the vertex $r$ has no outgoing edges. Therefore, the only variable of the form $x_{r, *}$ occurring in $P_{\pi}(x)$ is the variable $x_{r, \pi(r)}$, which divides each term of $P_{\pi}(x)$ exactly once. Define $Q_{\pi}(x) \triangleq P_{\pi}(x)/x_{r, \pi(r)}$. To prove \Cref{eq:submatching}, it then suffices to show
\begin{equation}\label{eq:qform}
\sum_{\pi | \pi(r) = c} Q_{\pi}(x) = \sum_{\pi | \pi(r) = i}Q_{\pi}(x).
\end{equation}

In order to prove \Cref{eq:qform}, first note that the LHS of this equation, i.e., the sum of $Q_\pi(x)$ over permutations $\pi\in S_n$ with $\pi(r)=c$, can be written as: 
\begin{equation}\label{eq:qalternate}
\sum_{\pi|\pi(r)=c}\prod_{i\in[n],i\neq r}x_{i, \pi(i)}\sum_{T\in\mathcal{T}_r(n)}\prod_{(u, v) \in T} x_{u, \pi(v)}=\sum_{\pi|\pi(r)=c}\sum_{T\in\mathcal{T}_r(n)}\prod_{(u, v) \in T} x_{u,\pi(u)}x_{u, \pi(v)},
\end{equation}
simply because each $u\in [n]\setminus\{r\}$ has exactly one outgoing edge in every arborescence $T\in \mathcal{T}_r(n)$. We will interpret the RHS of \Cref{eq:qalternate} as enumerating certain undirected bipartite graphs. Given a permutation $\pi$ with $\pi(r)=c$ and an $r$-rooted arborescence $T$, consider an undirected bipartite graph $G(\pi,T)$ on $2n$ vertices, with $n$ vertices on the left (labelled $1_L$ through $n_L$) and $n$ vertices on the right (labelled $1_R$ through $n_R$). The edges are constructed as follows (see \Cref{fig:bijection}):
\begin{itemize}
    \item for each $u \in [n] \setminus \{r\}$, add the edge $(u_{L}, \pi(u)_{R})$.
    \item for each directed edge $u \rightarrow v$ in $T$, add the edge $(u_{L}, \pi(v)_{R})$. 
\end{itemize}
Now, it is straightforward to verify that the summation in \Cref{eq:qalternate} can be written as
$$\sum_{\pi | \pi(r) = c} Q_{\pi}(x) = \sum_{\pi | \pi(r) = c}\sum_{T \in \mathcal{T}_{r}(n)}\prod_{(u_{L}, v_{R}) \in G(\pi, T)} x_{u, v}.$$
Next, we define a collection of bipartite graphs $\mathcal{G}_r$ on vertices $\{1_L,\ldots,n_L\}\cup \{1_R,\ldots,n_R\}$ for a fixed root $r$ -- which we call \emph{$r$-bi-trees} (see the definition below). We then claim there is a bijection between $(\pi,T)$ pairs in the above summation (with $\pi(r)=c$) and bipartite graphs $G'\in \mathcal{G}_r$  where $G'=G(\pi,T)$. If the claim holds, we have
$$\sum_{\pi | \pi(r) = c} Q_{\pi}(x) = \sum_{G'\in\mathcal{G}_r}\prod_{(u_{L}, v_{R}) \in G'} x_{u, v},$$
and since the RHS does not depend on the identity of $c$, it immediately implies \Cref{eq:qform}. 
\begin{definition}[$r$-bi-tree]
\label{def:rbitree}
For any root $r\in[n]$, an undirected bipartite graph $G$ on $2n$ vertices $\{1_L,\ldots,n_L\}$ and  $\{1_R,\ldots,n_R\}$ is an \textit{$r$-bi-tree} if it satisfies the following conditions:

\begin{enumerate}[label=(\roman*)]
    \item The vertex $r_{L}$ is an isolated vertex. 
    \item The remainder of the vertices (aside from $r_{L}$) belong to a single connected component.
    \item Each vertex $u_{L}$ (where $u \neq r$) on the left side has degree exactly equal to $2$. 
\end{enumerate}
\end{definition}






\begin{figure}[htb]
\centering
         \centering
        \includegraphics[trim={5.1cm 6cm 10.5cm 13cm},clip,width=1\textwidth]{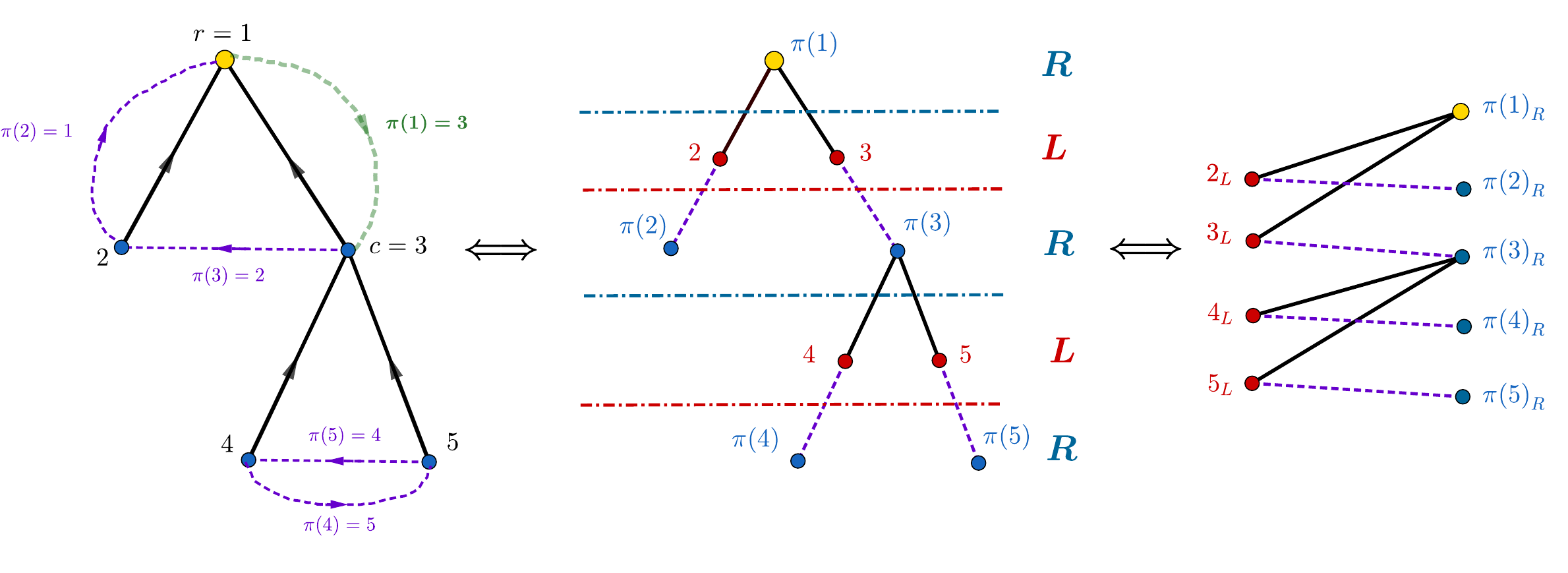}
\caption{ \label{fig:bijection}An example of the bijection between pairs $(\pi,T)\in S_n\times \mathcal{T}_r(n)$ with $\pi(r)=c$ \emph{(left hand side)} and $r$-bi-trees $G'\in \mathcal{G}_r$ \emph{(right hand side)}: $n=5$, root $r=1$, $\pi=(3,1,2,5,4)$, $c=\pi(r)=3$, and $T=\{2\rightarrow 1, 3\rightarrow 1, 5\rightarrow 3,4 \rightarrow 3\}$; solid black edges belong to $T$; dashed purple edges are matching edges corresponding to $\pi$ excluding the green dashed edge $\left(r,\pi(r)\right)$, i.e., $\{(u,\pi(u))\}_{u\in\{2,3,4,5\}}$; }
\end{figure}

We finish the proof by sketching why the above bijection claim holds in the following lemma. We postpone a more detailed proof of this lemma to \Cref{apx:matching} of the online supplement. 
\begin{lemma}[Bijection] 
\label{lemma:bijection}
For any $r,c\in[n]$, there exists a one-to-one correspondence between pairs $(\pi,T)\in S_n\times\mathcal{T}_r(n)$ where $\pi(r)=c$ and $r$-bi-trees $G'$ in $\mathcal{G}_r$ where $G'=G(\pi,T)$.
\end{lemma}
\begin{proof}[Proof sketch.] We prove the bijection in two parts: 

\vspace{1mm}
\noindent\emph{Part~(i)}: first, we claim $G(\pi,T)$ is an $r$-bi-tree. Notice that $G'=G(\pi,T)$ can be constructed from $(\pi,T)$ by a \emph{reverse} breadth-first search (BFS) walk on $T$ starting from root $r$, and then adding both edges $(u_L,\pi(v)_R)$ and $(u_L,\pi(u)_R)$ to $G'$ each time the walk moves from a vertex $v$ to a vertex $u$ (this is possible only if $u\rightarrow v$ is a directed edge in $T$). This step can alternatively be seen as adding a path of length $2$ from $\pi(v)_R$ to $\pi(u)_R$, passing through $u_L$, in $G'$. See \Cref{fig:bijection} (left to right) for a pictorial demonstration.  Now the claim can be proved as follows. As root $r$ has no outgoing edges in $T$, $r_L$ does not appear in any edges of $G'$ and remains isolated. Moreover, each $u\neq r$ is visited exactly once in the reverse BFS walk, which adds exactly two edges incident to $u_L$ in $G'$. Therefore, each $u_L$ for $u\neq r$ has degree $2$. Finally, $G'$ has no cycles, as we basically replace each edge in the undirected version of $T$ with a path of length 2 to construct $G'$. As the forest $G'$ has $2n-2$ edges, the remaining $2n-1$ vertices aside from $r_L$ should belong to a single connected component, which finishes the proof of our first claim. 

\vspace{1mm}
\noindent\emph{Part~(ii)}: second, we show the mapping $G(\pi,T)$ has an inverse. In other words, we propose an inverse mapping that given an $r$-bi-tree $G'$ uniquely returns a permutation $\pi$ (satisfying $\pi(r)=c)$ and $r$-rooted arborescence $T\in\mathcal{T}_r(n)$, so that $G(\pi,T)=G'$. To construct such a pair $(\pi,T)$, consider a BFS walk on the given undirected bipartite graph $G'$ starting from $c_R=\pi(r)_R$ (index the BFS tree layers by $0,1,2,\ldots$). We first show how to construct a permutation $\pi$ satisfying $\pi(r)=c$ from the walk. As $G'\setminus\{r_L\}$ is a single connected component, the BFS walk will visit all the vertices in $G'$ except for $r_L$. Moreover, in each odd layer of the BFS walk, it visits a left vertex $u_L$ with degree exactly $2$, as $G'$ is an $r$-bi-tree. Once the walk enters $u_L$, there is only one remaining indecent edge $(u_L,v_R)$ that can be added next to the BFS tree. Add this edge to the ``matching" $\pi$ by setting $\pi(v)=u$. At the end of the walk, the constructed $\pi$ (together with setting $\pi(r)=c$) gives a permutation as desired, simply because the BFS tree visits every right hand side vertex exactly once. Next, revisit the BFS walk and construct an arborescence $T$ by adding a directed edge $u\rightarrow v$ to $T$ for every edge $(u_L,\pi(v)_R)$ going from an even layer to an odd layer of the BFS tree (or equivalently, for every path of length $2$ in the BFS walk from an even layer vertex $\pi(v)_R$ to another even layer vertex $\pi(u)_R$ add a directed edge $u\rightarrow v$ to $T$). See \Cref{fig:bijection} (right to left) for a pictorial demonstration. As the BFS tree visits every vertex on the right side of $G'$ exactly once and $\pi$ is a permutation, the directed graph $T$ will be an arborescence rooted at $\pi^{-1}(c)=r$, as desired. Moreover, from the construction it is clear that a reverse BFS walk as described in the Part~(i) of the proof using $(\pi,T)$ will return $G'$. Hence, $G'=G(\pi,T)$. 
\end{proof}

\section{Necessary Conditions for Factories for Polytopes}\label{sec:necessary}

We now begin our exploration of the general combinatorial Bernoulli factory problem: for which polytopes $\P \subseteq [0, 1]^n$ does there exist a Bernoulli factory for $\P$? In this section we provide a necessary condition: any such $\P$ must be the intersection of $[0, 1]^n$ with an affine subspace. Recall that an affine subspace $\H$ of $\R^d$ is a set of points $x$ satisfying $Wx = b$ for some full-rank $k$-by-$d$ matrix $W$ and $b \in \R^k$ (in this case, we say the \textit{codimension} of $\H$ is $k$). 

\begin{theorem}\label{thm:necessary}
Let $\P \subseteq [0, 1]^n$ be a polytope such that $\P \cap (0, 1)^n \neq \emptyset$. If $\P$ is not of the form $\P = [0,1]^n \cap \H$ for some affine subspace $\H$ then no Bernoulli factory for $\P$ exists. 
\end{theorem}

Since (non-strong) Bernoulli factories for $\P$ are only required to work for $x \in \P \cap (0, 1)^n$, the constraint that $\P \cap (0, 1)^n \neq \emptyset$ is necessary. For strong Bernoulli factories that work for all $x \in \P$, we have the following stronger theorem.

\begin{theorem}\label{thm:necessary_strong}
Let $\P$ be a polytope. If $\P$ is not of the form $\P = [0, 1]^n \cap \H$ for some affine subspace $\H$, then no strong Bernoulli factory for $\P$ exists.
\end{theorem}

The full proofs of Theorems \ref{thm:necessary} and \ref{thm:necessary_strong} can be found in  \Cref{app:necessary} of the online supplement. In the remainder of this section, we provide a sketch of the main ideas in this proof. 

Before we proceed, it will prove illustrative to understand some of the obstacles to producing one-parameter Bernoulli factories for certain functions $f: [0, 1] \rightarrow [0, 1]$ (i.e., the classic Bernoulli factory setting studied in \citep{keane1994bernoulli}). Consider, for example, the function $f(x) = |x - 0.5|$. On first glance, since $f(x) \in [0, 1]$ for all $x \in [0, 1]$, it might appear possible to construct a one-bit Bernoulli factory $\mathcal{F}$ for $f$. However, this is impossible. One reason why is that since $f(0.6) = 0.1 > 0$, there must be some finite sequence of coin flips where $\mathcal{F}$ outputs $1$ (i.e., a leaf $\ell$ labelled $1$ in the tree for $\mathcal{F}$ where $\Pr[\mathcal{F}(0.6) \rightarrow \ell] > 0$). But this finite sequence of coin flips must also occur with positive probability when $x = 0.5$, so $f(0.5)$ must also be strictly positive. In general, if any non-constant $f(x): [0, 1] \rightarrow [0, 1]$ achieves the value $0$ or $1$ on $(0, 1)$, this argument shows it is not possible to construct a Bernoulli factory for $f$.

\begin{wrapfigure}{r}{0.25\textwidth}
\centering
\begin{tikzpicture}[scale=.85]

\draw[line width=1pt] (0,-1)--(0,2)--(3,2)--(3,-1)--cycle;
\draw[line width=1pt, color=blue, fill=blue!20!white] (0,-1)--(0,2)--(3,2)--cycle;
    \node[circle,fill,inner sep=1.5pt] at (1.5,.5) {};
  \node[circle,fill,inner sep=1.5pt] at (1,1) {};
  \node[circle,fill,inner sep=1.5pt] at (0,2) {};
  \node at (1.6,.2) {$x_1$};
  \node at (.8,1.3) {$x_2$};
  \node at (-.2,2.2) {$v$};
\end{tikzpicture}
\caption{The factory should output $v$ at $x_2$ but not at $x_1$.}
\label{fig:2d_fac}
\end{wrapfigure}
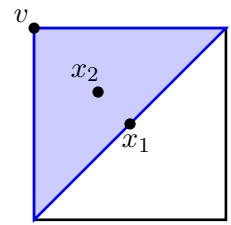

Similar obstructions appear when designing Bernoulli factories for polytopes. Consider, for example, the polytope $\P \subseteq [0, 1]^2$ with vertices $(0, 0)$, $(0, 1)$, and $(1, 1)$ (see Figure \ref{fig:2d_fac}) and assume to the contrary we have a Bernoulli factory $\mathcal{F}$ for $\P$. Let $v = (0, 1)$. Note that for a point $x_2$ in the interior of $\P$, $\mathcal{F}(x_2)$ must output $v$ with positive probability (since $x_2$ cannot be written as a convex combination of the two other vertices). Similarly, for a point $x_1$ in the middle of the edge connecting $(0, 0)$ and $(1, 1)$, $\mathcal{F}(x_1)$ must output $v$ with \textit{zero} probability. But these two statements are incompatible; if $\mathcal{F}(x_1)$ outputs $v$ with positive probability, there is some leaf $\ell$ labelled $v$ in the protocol tree for $\mathcal{F}$ such that $\Pr[\mathcal{F}(x_2) \rightarrow \ell] > 0$. But $\Pr[\mathcal{F}(x) \rightarrow \ell]$ is just a Bernstein monomial in $x$, so if $\Pr[\mathcal{F}(x_2) \rightarrow \ell] > 0$ for an $x_2 \in (0, 1)^2$, it follows that $\Pr[\mathcal{F}(x_1) \rightarrow \ell] > 0$ (since $x_1$ also lies in $(0, 1)^2$). This means no Bernoulli factory for $\P$ can exist.

The general proof of Theorem \ref{thm:necessary} proceeds along these lines. We formalize this by looking at the \textit{open faces} of $\P$. A face of $\P$ is a set of points in $\P$ which maximize a linear functional. The polytope $\P$ in Figure \ref{fig:2d_fac} contains $7$ faces: one 2-dimensional face (all of $\P$), three 1-dimensional faces (the edges of $\P$) and three 0-dimensional faces (the vertices of $\P$). The faces of $\P$ form a lattice; an open face is the set of points that belong to some face in $\P$ but no sub-faces (e.g. the interior of $\P$).

We begin by showing that if there are two different open faces of $\P$ contained in $(0, 1)^n$, then there is no Bernoulli factory for $\P$ (Lemma \ref{lem:char1} of the online supplement). For example, the $\P$ in Figure \ref{fig:2d_fac} has two open faces that are subsets of $(0, 1)^n$: the 2-dimensional open face $\mathrm{int}(\P)$ and the 1-dimensional open edge between $(0, 0)$ and $(1, 1)$. The proof of this Lemma is similar to the reasoning above; if we have two points $x_1, x_2$ in the interior of $(0, 1)^n$ that belong to different open faces of $\P$, we can show that there is some vertex which must occur with positive probability in $\mathcal{F}(x_1)$ but with zero probability in $\mathcal{F}(x_2)$. This, however, is impossible for the same reason as above (since a non-zero Bernstein monomial is positive everywhere on $(0, 1)^n$).

We then show that if the unique open face of $\P$ contained in $(0, 1)^n$ is the interior of $\P$, $\P$ is the intersection of $[0, 1]^n$ and an affine space (Lemma \ref{lem:linear} of the online supplement). To see this, we prove the contrapositive -- assume $\P$ is not the intersection of $[0, 1]^n$ with an affine subspace. Then look at the affine span $\mathcal{H}$ of $\P$, and let $\cQ$ be the polytope formed by the intersection of $[0, 1]^n$ and $\mathcal{H}$. We now know $\P$ is strictly contained in $\cQ$ -- using this, we can show that there is a boundary face of $\P$ in the interior of $\cQ$. But then there are two open faces of $\P$ in the interior of $(0, 1)^n$: this boundary face and the interior of $\P$. Combining Lemmas \ref{lem:char1} and \ref{lem:linear}, we arrive at Theorem \ref{thm:necessary}. The proof of Theorem \ref{thm:necessary_strong} proceeds similarly -- it suffices to look at the smallest face of $[0, 1]^n$ containing $\P$.

\section{Bernoulli Factories for Generic Polytopes}\label{sec:generic}

We will start by building Bernoulli factories for polytopes of the form $\P = [0,1]^n \cap \H$ for generic affine spaces $\H$. Later, we will extend this construction to non-generic spaces. A  $k$-dimensional affine subspace $\H$ can be written in the form
$$\H = \{x \in \R^n \mid W x = b\}$$
where $W$ is a $k \times n$-matrix of rank $k$ and $b \in \R^k$. We will let $w^{i}$ denote the $i$-th column of $W$. Given a subset $S\subseteq [n]$ we define the matrix $W_S$ to be the submatrix formed by the columns of $W$ indexed by elements in $S$.

\paragraph*{Generic subspaces}
An affine subspace $\H$ is said to be \emph{generic} if for each subset $S$ of size $k$ such that $W_S$ is non-singular and for each subset $B \subseteq [n]\setminus S$ the solution $W_S^{-1}(b- \sum_{i\in B} w^i)$ has no coordinates in $\{0,1\}$. Equivalently, $\H$ is generic if every vertex in $\P$ has exactly $k$ coordinates in the open interval $(0,1)$.

It is easy to check that for any fixed matrix $W$, the set of $b$ such that $Wx = b$ is non-generic forms a set of measure zero. So, by slightly perturbing $b$ it is always possible to obtain a generic subspace from a non-generic one. Many subspaces of interest in combinatorial optimization (e.g. $k$-subset, matchings, flows) are non-generic. We will later study these spaces as limits of generic affine spaces.

\paragraph*{Vertices and partitions} It is useful to represent each vertex of the polytope $\P$ with a partition of the set $[n]$ into three parts, indicating which coordinates of the vertex are equal to $0$, $1$, or lie in the open interval $(0, 1)$. We define the set of relevant partitions as follows:
$$\Part_{[n],k} \triangleq \{ (A,S,B) \mid A \cup S \cup B = [n], \abs{S} = k \text{ and } \abs{A} + \abs{S} + \abs{B} = n \}$$
We define the set of \emph{valid} partitions for the polytope $\P$ as:
$$\U \triangleq \{ (A,S,B) \in \Part_{[n],k} \mid \det W_S \neq 0 \text{ and } W_S^{-1}(b- \textstyle\sum_{i\in B} w^i) \in (0,1)^k \}$$

For polytopes formed from generic subspaces, there is a bijective mapping between vertices $v \in V$ and valid partitions $\pi \in \U$. Given a valid partition $\pi = (A,S,B)$ consider the vertex $v \in \P$ such that $v_i = 0$ for $i\in A$, $v_i=1$ for $i \in B$ and $v_S = W_S^{-1}(b- \sum_{i\in B} w^i)$ (where $v_S$ is a shorthand for the coordinates of $x$ corresponding to indices in $S$). Similarly, given a vertex $v \in V$ we can represent it by the partition $\pi = (A,S,B)$ where $A$ corresponds to the indices where $v_i =0$, $B$ corresponds to the indices where $v_i = 1$ and $S$ corresponds to the remaining indices.

Given a partition $\pi \in \U$ we will let $v^{\pi} \in V$ denote the vertex associated with this partition; likewise, given a vertex $v \in V$, we will let $\pi^{v} \in \U$ be the partition corresponding to this vertex. We write $A_{\pi}, S_{\pi}$ and $B_{\pi}$ to refer to the subsets in the partition $\pi$.

\paragraph*{Factory Construction} We will construct a Bernoulli factory for $\P$ by defining a Bernstein polynomial for each partition $\pi = (A,S,B) \in \Part_{[n],k}$:
\begin{equation}\label{eq:generic_factory}
P_{\pi}(x) \triangleq \abs{\det W_{S}} \cdot \prod_{i \in A} (1-x_i) \cdot \prod_{i \in B} x_i \cdot \prod_{i \in S} x_i (1-x_i)
\end{equation}
Now, for each vertex $v \in V$ we define $P_v(x)$ to be the polynomial associated with the corresponding partition $\pi^{v} \in \U$, i.e. $P_{v}(x) = P_{\pi^{v}}(x)$. At this point it is useful to note that for constructing the factory we only need $P_\pi$ for $\pi \in \U$, but we define the polynomials more generally since they will be useful in the proof.

\begin{theorem}\label{thm:valid_factory}
For a generic affine subspace $\H$, the Bernoulli race over Bernstein polynomials given by $P_{v}(x) = P_{\pi^{v}}(x)$ is a strong Bernoulli factory for $\P = [0,1]^n \cap \H$.
\end{theorem}

\subsection{Proof of Theorem \ref{thm:valid_factory}}

To show the race over our Bernstein polynomials is a strong Bernoulli factory for $\P$, we need to check Conditions~\eqref{eq:not_vanish}
and \eqref{eq:factory} in Theorem \ref{thm:factory_construction}. 

\subsubsection{Checking Condition \eqref{eq:not_vanish}} 

We start with the easier condition, that $\sum P_{v}(x)$ does not vanish on $\P$. Fix an $x \in \P$ and let
$$A_x = \{i\mid x_i = 0\} \quad S_x = \{i \mid 0< x_i < 1\} \quad B_x = \{i \mid x_i = 1\}.$$ 
Write $x$ as a convex combination of vertices and pick any vertex $v$ with positive weight in this combination. Note that if $x_i = 0$, then $v_i = 0$; likewise, if $x_i = 1$, then $v_i = 1$ (since $\P \subseteq [0, 1]^n$). It follows that if vertex $v$ corresponds to partition $\pi \in \U$, then
$$A_x \subseteq A_\pi \qquad S_\pi \subseteq S_x \qquad B_x \subseteq B_\pi.$$ 
Now, observe that
$$ P_\pi(x) =  \abs{\det W_{S_\pi}} \cdot \prod_{i \in A_\pi} (1-x_i) \cdot \prod_{i \in B_\pi} x_{i} \cdot \prod_{i \in S_\pi} x_{i} (1-x_{i}) = \abs{\det W_{S_\pi}} \prod_{i \in S_\pi} x_{i} (1-x_{i})> 0.$$
It follows that $\sum_{\pi \in \U} P_\pi(x) > 0$.

\subsubsection{Rewriting Condition~\eqref{eq:factory}}
The interesting part of the proof is to show that Condition~\eqref{eq:factory} holds. Recall that Condition~\eqref{eq:factory} states that
$$\sum_{v\in V}P_{v}(x)(v-x) = 0$$
must hold for all $x \in \P$. Since $x$ and $v$ are $n$-dimensional vectors, this is a vector equation. We will check this condition for each coordinate. Fix a coordinate $j \in [n]$ and split the sum $\sum_{\pi}P_{\pi}(x)(v^{\pi}-x)$ over all partitions $\pi \in \U$ depending on whether $j$ belongs to $A_\pi$, $B_\pi$ or $S_\pi$:
\begin{equation}\label{eq:coord_factory}
\sum_{\pi \in \U \mid j \in A_\pi} -x_j P_{\pi}(x) + \sum_{\pi \in \U \mid j \in B_\pi}  (1-x_j) P_{\pi}(x)  + \sum_{\pi \in \U \mid j \in S_\pi}  (
v_j^\pi -x_j) P_{\pi}(x) = 0
\end{equation}

We will now rewrite each of the terms below as sums over partitions in $\Part_{[n]\setminus j, k}$ (i.e. partitions of the set $[n]\setminus j$ into three parts $(A, B, S)$ where $S$ has $k$ elements).

\paragraph*{First term of \Cref{eq:coord_factory}} Given a partition $\pi = (A,S,B) \in \U$ with $j \in A$ consider a partition $\pi' = (A \setminus j, S, B)$. This establishes a bijective mapping between $\{\pi \in \U \mid j \in A_{\pi}\}$ and the set:
$$\U_A^j \triangleq \{ (A',S',B') \in \Part_{[n]\setminus j,k}  \mid \det W_{S'} \neq 0 \text{ and } W_{S'}^{-1}(b- \textstyle\sum_{t\in B'} w^t) \in (0,1)^k \}$$
which allows us to rewrite the first term in \Cref{eq:coord_factory} as follows:

\begin{equation}\label{eq:a}
\sum_{\pi \in \U| j \in A_\pi} -x_j P_{\pi}(x) = \sum_{\pi' \in \U_A^j} - x_j(1-x_j) \cdot P_{\pi'}(x)
\end{equation}
Here we define $P_{\pi'}(x)$ analogously to the definition of $P_{\pi}(x)$ in \eqref{eq:generic_factory}. Observe that $P_{\pi'}(x)$ does not have any terms depending on $x_j$; it is a polynomial in the $(n-1)$ other variables.

\paragraph*{Second term of \Cref{eq:coord_factory}}  Similarly for the second term, we can establish a bijective mapping between $\{\pi \in \U \mid j \in B_{\pi}\}$ and the set:
$$\U_B^j \triangleq \{ (A',S',B') \in \Part_{[n]\setminus j,k}  \mid \det W_{S'} \neq 0 \text{ and } W_{S'}^{-1}(b- w^j -\textstyle\sum_{t\in B'} w^t) \in (0,1)^k \}.$$
This allows us to rewrite:
\begin{equation}\label{eq:b}
\sum_{\pi \in \U| j \in B_\pi} (1-x_j) P_{\pi}(x) = \sum_{\pi' \in \U_B^j} x_j(1-x_j) P_{\pi'}(x).
\end{equation}

\paragraph*{Last term of \Cref{eq:coord_factory}}

Let's first examine the term $v^\pi_j-x_j$ in the last expression of \Cref{eq:coord_factory}. For this, it is useful to establish a bit of additional notation. Given a set $S$ of size $k$, recall that the matrix $W_S$ is the square matrix formed by taking the columns with indices in $S$ (in increasing order of the indices). Given coordinates $j \in S$ and $i \notin S$ we will define $W_{S[j\rightarrow i]}$ to be the matrix formed by replacing column $w^j$ by $w^i$. For example, if $S = \{2,3,5,7\}$ then:
$$W_S = [w^2 \text{ } w^3 \text{ } w^5 \text{ } w^7 ] \quad \text{and} \quad W_{S[5 \rightarrow 11]} = [w^2 \text{ } w^3 \text{ } w^{11} \text{ } w^7 ] $$
Note that the order where the $i$th column is inserted matters. With that, we are ready to state the next lemma:

\begin{lemma}
If $x \in \P$ and $v^{\pi}$ is a vertex corresponding to partition $\pi = (A,S,B)$ then for any coordinate $j \in S$ we have that
$$v^{\pi}_j - x_j = \sum_{i \in A}  \frac{\det W_{S[ j \rightarrow i]}}{\det W_{S}}  x_i - \sum_{i \in B}  \frac{\det W_{S[ j \rightarrow i]}}{\det W_{S}} (1-x_i).$$
\end{lemma}

\begin{proof}
We can write the $S$-components of $v$ as:
$$v_{S} = W_{S}^{-1}\left(b - \sum_{i \in B} w^i\right).$$
Since $x \in \P$ we know that
$$b = \sum_i w^i x_i = W_{S} x_{S} + \sum_{i \in A \cup B} w^i x_i.$$
Replacing this in the expression above we get:
$$v_{S} = x_{S} + W_{S}^{-1}\left(\sum_{i \in A} w^i \cdot x_i - \sum_{i \in B} w^i (1-x_i) \right)$$
Since $j \in S$ we can look at the $j$th component of the expression above. Observe that the $j$-th component of $W_{S}^{-1} w^i$ can be obtained via Cramer's rule and is given by
$$ [W_{S}^{-1} w^i]_j = \frac{\det W_{S[j\rightarrow i]}}{\det W_{S}}. $$
\end{proof}

The previous lemma allows us to rewrite the last term in \Cref{eq:coord_factory} as follows:
\begin{equation}\label{eq:c0} 
\begin{aligned}
\sum_{\pi \in \U \mid j \in S_\pi}  (
v_j^\pi -x_j) P_{\pi}(x) = & \sum_{i \neq j}  \sum_{\stackrel{\pi \in \U \mid j \in S_\pi,}{i \in A_\pi}} \frac{\det W_{S_\pi[ j \rightarrow i]}}{\det W_{S_\pi}} \cdot x_i P_\pi(x) \\ & - \sum_{i \neq j} \sum_{\stackrel{\pi \in \U \mid j \in S_\pi,}{ i \in B_\pi}}  \frac{\det W_{S_\pi[ j \rightarrow i]}}{\det W_{S_\pi}} \cdot (1-x_i)  P_\pi(x) 
\end{aligned}
\end{equation}

As before we will rewrite each of these terms as sums of partitions over $[n]\setminus j$. Starting with the first term, observe that for a fixed $i \neq j$ we can establish a bijective mapping between $\{\pi \in \U \mid j \in S_\pi \text{ and } i \in A_\pi\}$ and the set:
$$\U_A^i \triangleq \{ (A',S',B') \in \Part_{[n]\setminus j,k} \mid i \in S',~ \det W_{S'[i \rightarrow j]} \neq 0, \text{ and } W_{S'[i \rightarrow j]}^{-1}(b-  \textstyle\sum_{i\in B'} w^i) \in (0,1)^k \}$$
by mapping $\pi = (A,S,B)$ to $\pi' = (A \setminus i, S \cup i \setminus j, B)$.  We now note that:
$$x_i P_{\pi}(x) =  \frac{\abs{\det W_S}}{\abs{\det W_{S[j \rightarrow i]}}} x_j (1-x_j) P_{\pi'}(x).$$
If we define $\sigma_{ij}(S)$ for a set $S$ of size $k$ with $i \in S$ and $j \notin S$ as
$$\sigma_{ij}(S) \triangleq \sign \left( \frac{\det W_{S[ i \rightarrow j]}}{\det W_{S}} \right) \in \{-1, 0, +1\}$$
then we can rewrite the first term in \Cref{eq:c0} in the form
\begin{equation}\label{eq:ca}
   \sum_{\pi \in \U \mid j \in S_\pi, i \in A_\pi} \frac{\det W_{S_\pi[ j \rightarrow i]}}{\det W_{S_\pi}} \cdot x_i P_v(x)  = x_j (1-x_j)  \sum_{\pi' \in \U_A^i} \sigma_{ij}(S_{\pi'}) \cdot P_{\pi'}(x)
\end{equation}

Similarly, for the second term of \Cref{eq:c0} we can establish a bijective mapping between $\{ \pi \in \U \mid j \in S_\pi, i \in B_\pi\}$ and
$$\U_B^i \triangleq \{ (A',S',B') \in \Part_{[n]\setminus j,k}  \mid i \in S',~ \det W_{S'[i \rightarrow j]} \neq 0, \text{ and } W_{S'[i \rightarrow j]}^{-1}(b- w^j - \textstyle\sum_{i\in B'} w^i) \in (0,1)^k \}$$
by mapping $\pi = (A,S,B)$ to $\pi' = (A, S \cup i \setminus j, B  \setminus i)$.  Again, note that:
$$(1-x_i) P_{\pi}(x) =  \frac{\abs{\det W_S}}{\abs{\det W_{S[j \rightarrow i]}}} x_j (1-x_j) P_{\pi'}(x)$$
which allows us to write:
\begin{equation}\label{eq:cb}
 - \sum_{\pi| j \in S_\pi, i \in B_\pi} \frac{\det W_{S_\pi[ j \rightarrow i]}}{\det W_{S_\pi}} \cdot (1-x_i) P_\pi(x)  = x_j (1-x_j)  \sum_{\pi' \in \U_B^i} - \sigma_{ij}(S_{\pi'}) P_{\pi'}(x) 
\end{equation}

\paragraph*{Combining the terms} We have now rewritten all the terms of \Cref{eq:coord_factory} as sums of polynomials defined over partitions of $[n]\setminus j$. Combining Equations~\eqref{eq:a}, \eqref{eq:b}, \eqref{eq:c0}, \eqref{eq:ca} and \eqref{eq:cb}, we can rewrite  \eqref{eq:coord_factory} as
$$
 \sum_{\pi' \in \U_A^j} - P_{\pi'}(x) + \sum_{\pi' \in \U_B^j}  P_{\pi'}(x) + \sum_{i \neq j} \left[ \sum_{\pi' \in \U_A^i} \sigma_{ij}(S_{\pi'}) \cdot P_{\pi'}(x) -  \sum_{\pi' \in \U_B^i}  \sigma_{ij}(S_{\pi'}) \cdot P_{\pi'}(x) \right] = 0
$$
after cancelling all $x_j(1-x_j)$ terms. Our main goal is to prove this identity. It is useful to group together all partitions $\pi'$ for which $S_{\pi'}$ is the same. We will then show the following lemma:

\begin{lemma}\label{lemma:identity}
For any fixed $S' \subseteq [n] \setminus j$ with $\abs{S'} = k$ the following is an identity:
\begin{equation}\label{eq:identity}
 \sum_{\stackrel{\pi' \in \U_A^j}{S_{\pi'}=S'}} - P_{\pi'}(x) + \sum_{\stackrel{\pi' \in \U_B^j}{S_{\pi'}=S'}}  P_{\pi'}(x) + \sum_{i \neq j}  \sum_{\stackrel{\pi' \in \U_A^i}{S_{\pi'}=S'}} \sigma_{ij}(S') P_{\pi'}(x)  - \sum_{i \neq j} \sum_{\stackrel{\pi' \in \U_B^i}{S_{\pi'}=S'}}  \sigma_{ij}(S') P_{\pi'}(x) = 0.
\end{equation}
\end{lemma}

We will prove Lemma \ref{lemma:identity} by showing that each term $P_{\pi'}(x)$ appears twice in the expression, once with a positive sign and one with a negative sign. One nice aspect of focusing on a fixed $S'$ is that the magnitude of all leading coefficients in the Bernstein monomials are the same, so we need only worry about signs of these coefficients.

Interestingly, the proof that will follow will be geometric and will be based on decompositions of zonotopes.

\subsubsection{Partitions and Zonotopes}

A \emph{zonotope} is a polytope formed by the Minkowski sum of line segments. In other words, given vectors $w^1,\hdots, w^k$ we will define their associated zonotope as:
$$\Zon(w^1, \hdots, w^k) \triangleq \{ w^1 x_1 + \hdots + w^k x_k \mid x_i \in [0,1] \} $$
and the open zonotope as $\Zon^0(w^1, \hdots, w^k)$ as the interior of $\Zon(w^1, \hdots, w^k)$. Whenever $\det [w^1, \hdots, w^k] \neq 0$, this is given by:
$$\Zon^0(w^1, \hdots, w^k) \triangleq \{ w^1 x_1 + \hdots + w^k x_k| x_i \in (0,1) \} $$
Otherwise $\Zon^0(w^1, \hdots, w^k)$ is empty. Note that since each column vector $w^i \in \R^k$, these zonotopes are subsets of $\R^k$.

We can now rewrite the sets $\U_A^j$, $\U_B^j$, $\U_A^i$ and $\U_B^i$ in terms of membership in certain zonotopes. Since we are focusing on $S'$, let us focus on only the partitions that have $S'$. Let:
$$\Part_{[n]\setminus j, k}(S') \triangleq  \{ \pi' \in \Part_{[n]\setminus j, k}| S_{\pi'} = S' \} $$
$$\U_A^j(S') \triangleq \U_A^j \cap \Part_{[n]\setminus j, k}(S')$$
and similarly for the other sets. We will refer to the columns in $S'$ as $w^1, \hdots, w^k$ (i.e., $W_{S'} = [w^1, \hdots, w^k]$). Additionally, we will assume $\det W_{S'} \neq 0$ (otherwise Lemma \ref{lemma:identity} is trivial). Finally,  associate with each partition $\pi'$ the following vector:
$$q(\pi') \triangleq b - \sum_{t \in B_{\pi'}} w^t$$

We can now write our sets in terms of membership of $q(\pi')$ in a corresponding zonotope. In particular, we have that:
\begin{equation}\label{eq:ua}
\begin{aligned}
& \U_A^j(S') = \{ \pi' \in \Part_{[n]\setminus j, k}(S') \mid  q(\pi') \in Z_A^j\}, & & Z_A^j \triangleq \Zon^0(w^1,\hdots, w^k)  \\
& \U_B^j(S') = \{ \pi' \in \Part_{[n]\setminus j, k}(S') \mid  q(\pi') \in Z_B^j\}, & & Z_B^j \triangleq w^j + \Zon^0(w^1,\hdots, w^k) 
\end{aligned}
\end{equation}

For $\U_A^i(S')$ and $\U_B^i(S')$ we have the condition that $i \in S'$. Hence $\U_A^i(S') = \U_B^i(S') = \emptyset$ if $i \notin S'$ or $W_{S'[i \rightarrow j]} = 0$ and otherwise:
\begin{equation}\label{eq:ua2}
\begin{aligned}
& \U_A^i(S') = \{ \pi' \in \Part_{[n]\setminus j, k}(S') \mid  q(\pi') \in Z_A^i\}, & & Z_A^i \triangleq \Zon^0(w^1,\hdots, w^{i-1}, w^j, w^{i+1}, \hdots w^k)  \\
& \U_B^i(S') = \{ \pi' \in \Part_{[n]\setminus j, k}(S') \mid  q(\pi') \in Z_B^i\}, & & Z_B^i \triangleq w^i + \Zon^0(w^1,\hdots, w^{i-1}, w^j, w^{i+1}, \hdots w^k) 
\end{aligned}
\end{equation}

When we loop over all partitions in $\pi' \in \Part_{[n]\setminus j, k}(S')$, we will observe that its corresponding $q(\pi')$ either belongs to none of these zonotopes or to exactly two. In the latter case, we will show that it gets assigned opposite signs. To build intuition, we start with the case where $k=2$, where we can geometrically visualize the proof.

\subsubsection{Geometric illustration of Lemma \ref{lemma:identity} for $k=2$}\label{sec:illustration}
Assume that $S' = \{1,2\}$ with $j \notin S'$. The partitions $\U_A^i(S')$ and $\U_B^i(S')$ for $i\neq 1,2$ are empty and can be ignored. We are then left with the following terms:
\begin{equation}
 -\sum_{\U_A^j(S')} P_{\pi'}(x)
 + \sum_{\U_B^j(S')} P_{\pi'}(x)
 + \sigma_1 \sum_{\U_A^1(S')} P_{\pi'}(x)
 + \sigma_2 \sum_{\U_A^2(S')} P_{\pi'}(x)
 - \sigma_1 \sum_{\U_B^1(S')} P_{\pi'}(x)
 - \sigma_2 \sum_{\U_B^2(S')} P_{\pi'}(x) = 0.
\end{equation}
Here we abbreviate $\sigma_{ij}(S')$ as $\sigma_i$ since $j$ and $S'$ are fixed. For now, assume that $\det[w^j w^2]$ and $\det[w^1 w^j]$ are non-zero such that $\sigma_1, \sigma_2 \in \{-1,+1\}$. We will consider 4 cases depending on the sign patterns of $(\sigma_1, \sigma_2)$.

We now can go over all partitions $\pi' \in \Part_{[n]\setminus j, k}(S')$ and assign them a positive sign whenever $q(\pi')$ falls in a region with positive sign or a negative sign if they fall in a region with negative sign. The sign will depend on the sign patterns of $\sigma_1, \sigma_2$. For $k=2$ it is instructive to look at each of the four sign patterns.

\paragraph*{Sign pattern $\boldsymbol{\sigma_1 = \sigma_2 = +1}$.}
  We have:
$$\sigma_1 = \sign\left( \frac{\det [w^j w^2]}{\det [w^1 w^2]} \right) \qquad \sigma_2 = \sign\left( \frac{\det [w^1 w^j]}{\det [w^1 w^2]} \right)$$
Geometrically, this means that the the sign of the angle from $w^1$ to $w^2$ (if the angle is in $(-\pi, \pi]$) is the same as the sign of the angle from $w^j$ to $w^2$ and the sign of the angle from $w^1$ to $w^j$. Figure \ref{fig:2d_pp} shows a configuration of such vectors.
\vspace{-3mm}
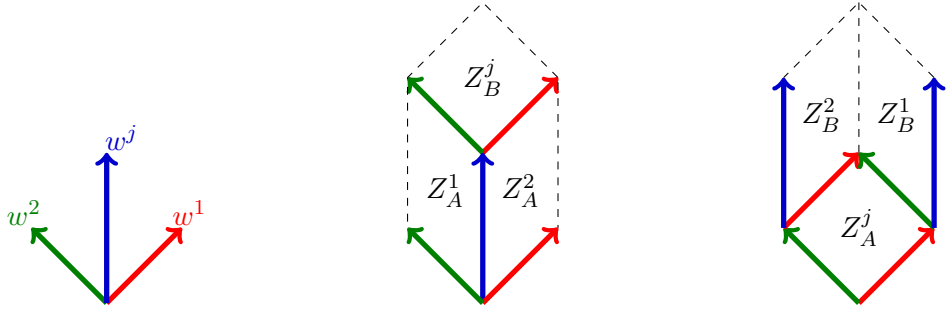
\begin{figure}[h]
\centering
\begin{tikzpicture}[scale=0.8]
  \draw [->, line width=2pt, color=red] (0,0) -- (1,1);
  \node at (1.1, 1.2) {\color{red} $w^1$};
  \draw [->, line width=2pt, color=black!50!green] (0,0) -- (-1,1);
  \node at (-1.1, 1.2) {\color{black!50!green} $w^2$};
\draw [->, line width=2pt, color=black!20!blue] (0,0) -- (0,2);
\node at (.2, 2.2) {\color{black!20!blue} $w^j$};

\begin{scope}[xshift=5cm]
  \draw [->, line width=2pt, color=red] (0,2) -- (1,3);
  \draw [->, line width=2pt, color=black!50!green] (0,2) -- (-1,3);
  \draw [dashed] (1,3) -- (0,4) -- (-1,3) -- (-1,1) -- (0,0) -- (1,1) -- cycle;
  \draw [->, line width=2pt, color=black!20!blue] (0,0) -- (0,2);
\draw [->, line width=2pt, color=red] (0,0) -- (1,1);
  \draw [->, line width=2pt, color=black!50!green] (0,0) -- (-1,1);
  \node at (0,3) {$Z_B^j$};
  \node at (-.5,1.5) {$Z_A^1$};
  \node at (.5,1.5) {$Z_A^2$};
\end{scope}

\begin{scope}[xshift=10cm]
  \draw [->, line width=2pt, color=red] (0,0) -- (1,1);
  \draw [->, line width=2pt, color=black!50!green] (0,0) -- (-1,1);
   \draw [->, line width=2pt, color=red] (-1,1) -- (0,2);
  \draw [->, line width=2pt, color=black!50!green] (1,1) -- (0,2);
  \draw [->, line width=2pt, color=black!20!blue] (1,1) -- (1,3);
\draw [->, line width=2pt, color=black!20!blue] (-1,1) -- (-1,3);
  \draw [dashed] (1,3) -- (0,4) -- (-1,3);
  \draw [dashed] (0,2) -- (0,4);
\node at (0,1) {$Z_A^j$};
  \node at (-.5,2.5) {$Z_B^2$};
  \node at (.5,2.5) {$Z_B^1$};
\end{scope}

\end{tikzpicture}
\caption{$\sigma_1 = \sigma_2 = +1$}
\label{fig:2d_pp}
\end{figure}

Under this sign pattern, the signs attributed to each region are the following:
$$Z_A^j (-) \quad Z_A^1 (+) \quad Z_A^2 (+) \quad 
Z_B^j (+) \quad Z_B^1 (-) \quad Z_B^2 (-) $$
It is simple to see in the picture that the pairwise intersection between the positive regions is disjoint. The same is true for the negative regions. Finally, their union generates the same set. In fact, both are tilings of the zonotope $\Zon(w^1, w^2, w^j)$
by smaller parallelograms formed by removing one of the vectors.

Before we proceed to other sign patterns, observe that it is not quite true that $Z_B^j \cup Z_A^1 \cup Z_A^2 = Z_A^j \cup Z_B^1 \cup Z_B^2$. The precise statement is that the union of their (topological) closures is the same (where $\bar X$ is the closure of $X$):
$$\bar Z_B^j \cup \bar Z_A^1 \cup \bar Z_A^2 = \bar Z_A^j \cup \bar Z_B^1 \cup \bar Z_B^2$$
This is, however, enough for our purposes since the $q(\pi)$ can never be in the boundary $Z^i_A$ or $Z^i_B$ due to our \emph{genericity} condition. To see this, observe that we can rewrite the definition of genericity $W_{S_\pi}^{-1}(b- \sum_{i \in B_\pi} w^i) \in (0,1)^k$ equivalently as $q(\pi) \in \Zon_0^j(w^1, \hdots, w^k)$.

\paragraph*{Sign pattern $\boldsymbol{\sigma_1 = +1, \sigma_2 = -1}$.} In Figure \ref{fig:2d_pm} we depict an example configuration of vectors $w^1$, $w^2$ and $w^j$ satisfying this sign pattern. Based on this sign pattern, the regions get assigned the signs:
$$Z_A^j (-) \quad Z_A^1 (-) \quad Z_A^2 (+) \quad 
Z_B^j (+) \quad Z_B^1 (+) \quad Z_B^2 (-) $$
Again we (visually) observe the same phenomenon:

\begin{figure}[h]
\centering
\begin{tikzpicture}[scale=0.8]
  \draw [->, line width=2pt, color=red] (0,0) -- (1,1);
  \node at (1.1, 1.2) {\color{red} $w^1$};
  \draw [->, line width=2pt, color=black!50!green] (0,0) -- (-1,1);
  \node at (-1.1, 1.2) {\color{black!50!green} $w^2$};
\draw [->, line width=2pt, color=black!20!blue] (0,0) -- (-2,0);
\node at (-2.2,.2) {\color{black!20!blue} $w^j$};

\begin{scope}[xshift=5cm]
  \draw [->, line width=2pt, color=red] (-2,0) -- (-1,1);
  \draw [->, line width=2pt, color=black!50!green] (-2,0) -- (-3,1);
  \draw [->, line width=2pt, color=black!20!blue] (0,0) -- (-2,0);
  \draw [dashed] (-1,1) -- (-2,2) -- (-3,1);
  \node at (-2,1) {$Z_B^j$};
  \draw [->, line width=2pt, color=red] (0,0) -- (1,1);
  \draw [->, line width=2pt, color=black!50!green] (1,1) -- (0,2);
  \draw [->, line width=2pt, color=black!20!blue] (1,1) -- (-1,1);
  \draw [dashed] (0,2) -- (-2,2);
  \node at (-.5,1.5) {$Z_B^1$};
  \node at (-.5,.5) {$Z_A^2$};
\end{scope}

\begin{scope}[xshift=10cm]
  \draw [->, line width=2pt, color=red] (0,0) -- (1,1);
  \draw [->, line width=2pt, color=black!50!green] (0,0) -- (-1,1);
  \draw [dashed] (1,1) -- (0,2) -- (-1,1);
  \node at (0,1) {$Z_A^j$};
    \draw [->, line width=2pt, color=black!20!blue] (0,0) -- (-2,0);
  \draw [dashed] (-2,0) -- (-3,1) -- (-1,1);
  \node at (-1.5,.5) {$Z_A^1$};
  \draw [->, line width=2pt, color=red] (-1,1) -- (0,2);
    \draw [->, line width=2pt, color=black!20!blue] (-1,1) -- (-3,1);
  \draw [dashed] (-3,1) -- (-2,2) -- (0,2);
  \node at (-1.5,1.5) {$Z_B^2$};
\end{scope}
\end{tikzpicture}
\caption{$\sigma_1 = +1,  \sigma_2 = -1$}
\label{fig:2d_pm}
\end{figure}
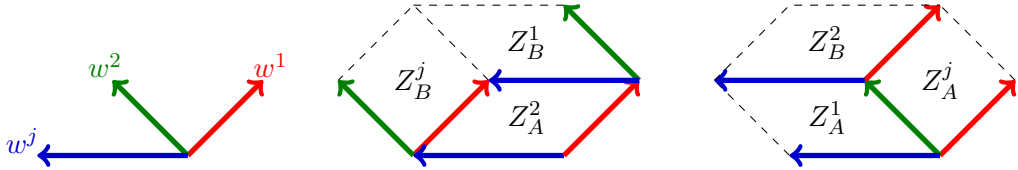
\vspace{-3mm}
\paragraph*{Sign pattern $\boldsymbol{\sigma_1 = -1, \sigma_2 = +1}$.} In Figure \ref{fig:2d_mp} we depict an example configuration of vectors $w^1$, $w^2$ and $w^j$ satisfying this sign pattern. Based on this sign pattern, the regions get assigned the signs:
$$Z_A^j (-) \quad Z_A^1 (+) \quad Z_A^2 (-) \quad 
Z_B^j (+) \quad Z_B^1 (-) \quad Z_B^2 (+) $$

\begin{figure}[h]
\centering
\begin{tikzpicture}[scale=0.8]
  \draw [->, line width=2pt, color=red] (0,0) -- (1,1);
  \node at (1.1, 1.2) {\color{red} $w^1$};
  \draw [->, line width=2pt, color=black!50!green] (0,0) -- (-1,1);
  \node at (-1.1, 1.2) {\color{black!50!green} $w^2$};
\draw [->, line width=2pt, color=black!20!blue] (0,0) -- (2,0);
\node at (2.2,.2) {\color{black!20!blue} $w^j$};
\begin{scope}[xshift=5cm]
  \draw [->, line width=2pt, color=red] (2,0) -- (3,1);
  \draw [->, line width=2pt, color=black!50!green] (2,0) -- (1,1);
  \draw [->, line width=2pt, color=black!20!blue] (0,0) -- (2,0);
  \draw [dashed] (1,1) -- (2,2) -- (3,1);
  \node at (2,1) {$Z^j_B$};
   \draw [->, line width=2pt, color=black!50!green] (0,0) -- (-1,1);
  \draw [dashed] (-1,1) -- (1,1);
  \node at (.5,.5) {$Z^1_A$};
   \draw [->, line width=2pt, color=black!20!blue] (-1,1) -- (1,1);
   \draw [->, line width=2pt, color=red] (-1,1) -- (0,2);
   \draw [dashed] (0,2) -- (2,2);
  \node at (.5,1.5) {$Z_B^2$};
\end{scope}
\begin{scope}[xshift=10cm]
  \draw [->, line width=2pt, color=red] (0,0) -- (1,1);
  \draw [->, line width=2pt, color=black!50!green] (0,0) -- (-1,1);
  \draw [dashed] (1,1) -- (0,2) -- (-1,1);
  \node at (0,1) {$Z^j_A$};
  \draw [->, line width=2pt, color=black!20!blue] (0,0) -- (2,0);
  \draw [dashed] (2,0) -- (3,1) -- (1,1);
  \node at (1.5,.5) {$Z_A^2$};
  \draw [->, line width=2pt, color=black!20!blue] (1,1) -- (3,1);
  \draw [->, line width=2pt, color=black!50!green] (1,1) -- (0,2);
  \draw [dashed] (3,1) -- (2,2) -- (0,2);
  \node at (1.5,1.5) {$Z_B^1$};
\end{scope}
\end{tikzpicture}
\caption{$\sigma_1 = -1,  \sigma_2 = +1$}
\label{fig:2d_mp}
\end{figure}
\vspace{-4mm}
\noindent\textbf{Sign pattern $\boldsymbol{\sigma_1 = -1, \sigma_2 = -1}$.} In Figure \ref{fig:2d_mm} we depict an example configuration of vectors $w^1$, $w^2$ and $w^j$ satisfying this sign pattern. Based on this pattern, the regions get assigned the following signs:
$$Z_A^j (-) \quad Z_A^1 (-) \quad Z_A^2 (-) \quad 
Z_B^j (+) \quad Z_B^1 (+) \quad Z_B^2 (+) $$
\vspace{-4mm}
\begin{figure}[h]
\centering
\begin{tikzpicture}[scale=0.8]
  \draw [->, line width=2pt, color=red] (0,0) -- (1,1);
  \node at (1.1, 1.2) {\color{red} $w^1$};
  \draw [->, line width=2pt, color=black!50!green] (0,0) -- (-1,1);
  \node at (-1.1, 1.2) {\color{black!50!green} $w^2$};
\draw [->, line width=2pt, color=black!20!blue] (0,0) -- (0,-2);
\node at (.2,-2.2) {\color{black!20!blue} $w^j$};

\begin{scope}[xshift=5cm]
  \draw [dashed] (0,-2) -- (1,-1) -- (1,1) -- (0,2) -- (-1,1) -- (-1,1) -- (-1,-1) -- cycle;
  \draw [dashed] (-1,-1) -- (0,0) -- (1,-1);
  \draw [dashed] (0,0) -- (0,2);
  \draw [->, line width=2pt, color=red] (0,-2) -- (1,-1);
  \draw [->, line width=2pt, color=red] (-1,1) -- (0,2);
  \draw [->, line width=2pt, color=black!50!green] (1,1) -- (0,2);
  \draw [->, line width=2pt, color=black!50!green] (0,-2) -- (-1,-1);

   \draw [->, line width=2pt, color=black!20!blue] (1,1) -- (1,-1);
    \draw [->, line width=2pt, color=black!20!blue] (-1,1) -- (-1,-1);

   \node at (0,-1) {$Z_B^j$};
   \node at (.5,.5) {$Z_B^1$};
   \node at (-.5,.5) {$Z_B^2$};
\end{scope}
\vspace{-5mm}
\begin{scope}[xshift=10cm]
  \draw [dashed] (0,-2) -- (1,-1) -- (1,1) -- (0,2) -- (-1,1) -- (-1,1) -- (-1,-1) -- cycle;
  \draw [dashed] (0,-2) -- (0,0) -- (1,1);
  \draw [dashed] (0,0) -- (-1,1);
  \draw [->, line width=2pt, color=red] (0,0) -- (1,1);
  \draw [->, line width=2pt, color=black!50!green] (0,0) -- (-1,1);
\draw [->, line width=2pt, color=black!20!blue] (0,0) -- (0,-2);
  \node at (0,1) {$Z_A^j$};
  \node at (-.5,-.5) {$Z_A^1$};
  \node at (.5,-.5) {$Z_A^2$};
\end{scope}
\end{tikzpicture}
\caption{$\sigma_1 = -1,  \sigma_2 = -1$}
\label{fig:2d_mm}
\end{figure}

\noindent\textbf{Patterns with $\boldsymbol{\sigma_1 = 0$ or $\sigma_2 = 0}$.} If either $\det[w^j w^2] = 0$ ($w^j$ is parallel to $w^2$) or $\det[w^1 w^j] = 0$ ($w^j$ is parallel to $w^1$) then we can recover these patterns as limits of the previous patterns. Note that both determinants can't be simultaneously zero (unless $w^j = 0$) since we assume $w^1$ and $w^2$ are not parallel. We depict what the pattern $\sigma_1 = +1$, $\sigma_2 = 0$ looks like in Figure \ref{fig:2d_pz}. The sign patterns become:
$$Z_A^j (-) \quad Z_A^1 (0) \quad Z_A^2 (+) \quad 
Z_B^j (+) \quad Z_B^1 (0) \quad Z_B^2 (-) $$
The regions $Z_A^1$ and $Z_B^1$ disappear since $\det[w^j w^2]=0$. We can see that even in these degenerate cases, we still obtain a tiling of the zonotope $\Zon(w^1, w^2, w^j)$. The remaining cases are analogous.
\vspace{-2mm}
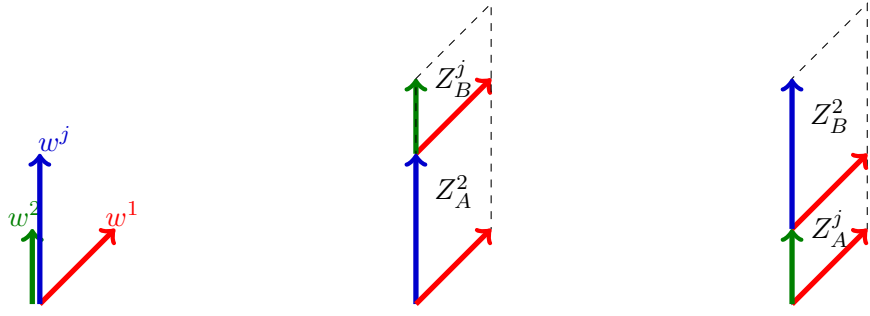
\begin{figure}[h]
\centering
\begin{tikzpicture}[scale=0.9]
\begin{scope}[]
  \draw [->, line width=2pt, color=red] (0,0) -- (1,1);
  \node at (1.1, 1.2) {\color{red} $w^1$};
  \draw [->, line width=2pt, color=black!50!green] (-.1,0) -- (-.1,1);
  \node at (-.2, 1.2) {\color{black!50!green} $w^2$};
\draw [->, line width=2pt, color=black!20!blue] (0,0) -- (0,2);
\node at (.2, 2.2) {\color{black!20!blue} $w^j$};
\end{scope}

\begin{scope}[xshift=5cm]
  \draw [->, line width=2pt, color=red] (0,2) -- (1,3);
  \draw [->, line width=2pt, color=black!50!green] (0,2) -- (0,3);
  \draw [dashed] (1,3) -- (1,4) -- (0,3) -- (0,0) -- (1,1) -- cycle;
  \draw [->, line width=2pt, color=black!20!blue] (0,0) -- (0,2);
\draw [->, line width=2pt, color=red] (0,0) -- (1,1);
  \node at (.5,3) {$Z_B^j$};
  \node at (.5,1.5) {$Z_A^2$};
\end{scope}

\begin{scope}[xshift=10cm]
  \draw [dashed] (1,2) -- (1,4) -- (0,3) -- (0,0) -- (1,1) -- cycle;
  \draw [->, line width=2pt, color=red] (0,0) -- (1,1);
  \draw [->, line width=2pt, color=black!50!green] (0,0) -- (0,1);
   \draw [->, line width=2pt, color=red] (0,1) -- (1,2);
\draw [->, line width=2pt, color=black!20!blue] (0,1) -- (0,3);
\node at (.5,1) {$Z_A^j$};
  \node at (.5,2.5) {$Z_B^2$};
\end{scope}

\end{tikzpicture}
\caption{$\sigma_1 = +1,  \sigma_2 = 0$}
\label{fig:2d_pz}
\end{figure}

\subsubsection{Proof of Lemma \ref{lemma:identity} for general $k$.} 
We now apply the geometric intuition developed in the last section to prove the general case of Lemma \ref{lemma:identity}. The main step will be to prove the following geometric lemma. 

\begin{lemma}[Tiling]\label{lemma:tiling}
The zonotopes with a positive sign in Lemma \ref{lemma:identity} have disjoint interior. Moreover, the union of their closures is the zonotope $\Zon(w^1, \hdots, w^k, w^j)$. The same is true for all the regions with a negative sign.
\end{lemma}

Intuitively, the (proof of the) Tiling Lemma says the following: there are $2(k+1)$ distinct terms in Lemma \ref{lemma:identity}; $(k+1)$ of these are positive and $(k+1)$ of these are negative. Each of these terms correspond to a zonotope (in fact, a parallelotope) in $q(\pi)$ space. The $(k+1)$ positive parallelotopes partition the zonotope $\Zon(w^1, \hdots, w^k, w^j)$, as do the $(k+1)$ negative parallelotopes.

Given the Tiling Lemma it is straightforward to show Lemma \ref{lemma:identity}:

\begin{proof}[Proof of Lemma \ref{lemma:identity}]
Given $\pi'$ in \Cref{eq:identity} then it appears in a given term iff $q(\pi')$ is in the corresponding zonotope. Since all the zonotopes corresponding to positive terms are non-overlapping, it can appear at most once with a positive term.

If it does appear with a positive term, then $q(\pi') \in \Zon(w^1, \hdots, w^k, w^j)$. By the fact that the instance is generic, $q(\pi')$ can't be in the boundary of any of the smaller zonotopes. Since the union of zonotopes with negative sign is also $\Zon(w^1, \hdots, w^k, w^j)$ then it must also appear in exactly one such zonotope. Hence it also appears exactly once with negative sign in \Cref{eq:identity}.
\end{proof}

We now devote the rest of the section to a proof of Lemma \ref{lemma:tiling}. The regions which get assigned a positive sign are i) $Z_B^j$, ii) $Z_A^i$ if $\sigma_i = +1$, and iii) $Z_B^i$ if $\sigma_i = -1$ (as in Section \ref{sec:illustration}, we suppress $j$ and $S'$ in $\sigma_{ij}(S')$). It is convenient to assign to $j$ a sign $\sigma_j = -1$ so to get a more uniform treatment of $j$ and $\{1, \hdots, k\}$; in particular, we are now simply considering all $Z_{B}^i$ where $\sigma_{i} = -1$, and all $Z_{A}^i$ where $\sigma_i = 1$. With this notation in mind, we first show that $Z_{B}^i$ and $Z_{A}^t$ must be disjoint if $\sigma_{i} \neq \sigma_{t}$. We further show that $Z_A^i$ and $Z_A^t$ must be disjoint if $\sigma_{i} =\sigma_{t}$.
\begin{lemma}\label{lemma:disjoint_pm}
Consider $i,t \in I \triangleq \{1,\hdots, k, j\}$ with $i\neq t$. If $\sigma_i = -\sigma_t$ then $Z_B^i$ and $Z_A^t$ are disjoint.
\end{lemma}
\begin{lemma}\label{lemma:disjoint_pp}
Consider $i,t \in I := \{1,\hdots, k, j\}$ with $i\neq t$. If $\sigma_i = \sigma_t$ then regions $Z_A^i$ and $Z_A^t$ are disjoint. The regions $Z_B^i$ and $Z_B^k$ are also disjoint
\end{lemma}
See \Cref{appendix:missingproofs} in the online supplement for the proof of these lemmas. With these two lemmas, we are ready to prove the Tiling Lemma:
\begin{proof}[Proof of Lemma \ref{lemma:tiling}]
Lemmas \ref{lemma:disjoint_pm} and \ref{lemma:disjoint_pp} show that the regions assigned positive sign are disjoint. It is straightforward to see that they are all contained in $\Zon(\{w^r\}_{r \in I})$ for $I = \{1,\hdots,k,j\}$ since it is simple to express a point in each region as $\sum_{r \in I} \lambda_r w_r$ with $\lambda_r \in [0,1]$. To show that their closures are exactly the zonotope it is enough to argue that their volumes sum up to the volume of $\Zon(\{w^r\}_{r \in I})$. Note that (since the zonotopes $Z_{A}^i$ and $Z_{B}^i$ are actually parallelotopes generated by $k$ vectors):
$$\Vol(Z_A^i) = \Vol(Z_B^i) = |\det W_{I \setminus \{i\}}|$$
where $W_{I \setminus i}$ is the matrix formed by columns $w^r$ for $r \in I \setminus i$. Finally the formula for computing the volume of a zonotope (see e.g. \citep{gover2010determinants}) is:
$$\Vol(\Zon(\{w^r\}_{r \in I})) = \sum_{T \subseteq I| \abs{T}=k} |\det W_{T}|$$ 
which is equal to the sum of volumes of the smaller parallelotopes.
\end{proof}

\section{Bernoulli Factories for Non-Generic Polytopes}\label{sec:non_generic}

In this section, we demonstrate how to obtain a factory for a non-generic polytope as a limit of factories for generic polytopes. We then look at the specific case of the $k$-subset polytope ($\{x \in [0, 1]^n \mid \sum_{i} x_i = k\}$), where we recover the statistical method known as Sampford sampling.
\begin{theorem}\label{thm:non_generic}
Consider a polytope of the form $\P = [0,1]^n \cap \H$, where $\H$ is a possibly non-generic affine subspace. Then there exist 
Bernstein polynomials $P_v(x)$ for each vertex $v$ of $\P$ such that the corresponding Bernoulli race over Bernstein polynomials is a Bernoulli factory for $\P$.
\end{theorem}

This gives us a valid Bernoulli factory that always terminates as long as all coins belong to the open set $(0,1)^n$, i.e., no coin is deterministically $0$ or $1$. In Section \ref{sec:impossibility_race_over_bernstein} of the online supplement we show that this limitation is unavoidable.

\begin{proof}[Proof of Theorem \ref{thm:non_generic}]
Let $\H = \{x \in \R^n \mid Wx = b\}$ for a $k \times n$ matrix $W$ and a vector $b \in \R^k$. For each $t = 1,2,3 \dots$ sample a vector $b_t \in \R^k$ uniformly from the ball of radius $1/t$ around $b$ and define $\H_t$ as the hyperplane:
$\H_t = \{x \in \R^k \mid Wx = b_t \}$. Since the set $\{b \in \R^k \mid Wx=b \text{ is generic} \}$ has measure zero, then $\H_t$ is generic almost surely for all $t$. Let $\P_t = [0,1]^n \cap \H_t$ and $\U_t$ be the set of valid partitions for polytope $\P_t$. Since $\U_t$ is a subset of the finite set $\Part_{[n],k}$, there are finitely many possibilities, so one of them must occur infinitely often. Passing to this subsequence if necessary, we can assume $\U_t$ is the same for all $t$. With this, observe that the Bernoulli factory is exactly the same for all polytopes $\P_t$ since the polynomials in Theorem \ref{thm:factory_construction} depend on $W$ and $\U$ but not directly on $b$. It follows that:
$$\sum_{\pi \in \U_t} P_\pi(x) \cdot (x - v^{t,\pi}) = 0, \quad \forall x \in \P_t$$
where $v^{t,\pi}$ is the vertex in $\P_t$ associated with partition $\pi$. Note that as $t \rightarrow \infty$ vertex $v^{t,\pi} \rightarrow v^\pi$ for the vertex $v^\pi$ in $\P$ associated with partition $\pi$. However, note that the correspondence is no longer 1-to-1, i.e., two different partitions $\pi_1, \pi_2$ may map to the same vertex in $\P$. 

Now take any $x \in \P$ and write it as a limit $x_t \rightarrow x$ with $x_t \in \P_t$. We know that: 
$\sum_{\pi \in \U_t} P_\pi(x) \cdot (x_t - v^{t,\pi}) = 0$. Taking the limit $t \rightarrow \infty$ we obtain $\sum_{\pi \in \U_t} P_\pi(x) \cdot (x - v^{\pi}) = 0$ which establishes condition \eqref{eq:factory}. Note that the polynomial associated with each vertex is:
\begin{equation}\label{eq:nongeneric_poly}
P_v(x) = \sum_{\pi \in \U_t| v = v^\pi} P_{\pi}(x).
\end{equation}
Finally, observe that condition \eqref{eq:not_vanish} is trivial for $x \in (0,1)^n$ since all Bernstein monomials are strictly positive at such points.
\end{proof}

\paragraph*{Warning} The proof of the previous theorem (and in particular \Cref{eq:nongeneric_poly}) give a recipe for constructing Bernoulli factories for non-generic hyperplanes. In practice, to construct a factory one can add a tiny perturbation to $b$, compute set $\U_t$ and use the formula in \Cref{eq:nongeneric_poly}.

One may be tempted to ignore the perturbation and try to apply \Cref{eq:nongeneric_poly} directly using $\U$ instead of $\U_t$. For this, one also needs to change $(0,1)^n$ to $[0,1]^n$ in the definition of $\U$ so that all vertices are represented. This approach, however, fails. One example is the $3$-by-$3$ perfect matching polytope. We implement\footnote{See the code in \url{https://gist.github.com/renatoppl/f9151d44e8ef798737e9ce75efbf0d1d}. The implementation is done in the computational algebra system SageMath.} the factory for matching both with and without the perturbation. With the perturbation we obtain a multiple\footnote{The polynomial obtained by the generic recipe has degree $n^2+2n-1$ while the one in Section \ref{sec:matching_factory} has degree $2n-1$. They differ by a factor of $\prod_{ij}(1-x_{ij})$. } of the polynomial  in Section \ref{sec:matching_factory}. If instead we do not add a perturbation, we obtain a family of polynomials that do not satisfy \Cref{eq:factory}.

\subsection{Sampling a $k$-subset (Sampford Sampling)}

We now show that for $k$-subset sampling, the recipe in Theorem \ref{thm:non_generic} recovers the procedure known as Sampford sampling \citep{sampford1967sampling}. Consider the polytope
\begin{equation}\label{eq:p_alpha}
\P_{\alpha,n} = \left\{x \in [0,1]^n| \sum_{i=1}^n x_i = \alpha \right\}
\end{equation}
for $\alpha \in (0,n)$. The vertices of $\P_{\alpha,n}$ are the vectors $v \in [0,1]^n$ having $k = \lfloor \alpha \rfloor$ coordinates equal to $1$, one coordinate equal to $\alpha - k$ and the remaining coordinates equal to 0.

If $\alpha$ is not an integer, then $\P_{\alpha,n}$ is generic and we can apply the construction in \eqref{eq:generic_factory} directly. The polynomial associated with vertex $v = (\alpha-k, 1, \hdots, 1, 0, \hdots, 0)$ is $P_v(x) = x_1(1-x_1) x_2 \hdots x_{k+1} (1-x_{k+2}) \hdots (1-x_n)$. Following the recipe for generic factories we obtain:
\begin{algorithm}[H]
\caption{{\sc Bernoulli Factory for $\P_{\alpha,n}$ for non-integer $\alpha$} (version 1)}
\label{alg:sampford_generic}
\begin{algorithmic} 
\State Pick a random vertex $v$
\State For each index such that $v_i = 1$, sample the $x_i$-coin and restart if it is $0$.
\State For each index such that $v_i = 0$, sample the $x_i$-coin and restart if it is $1$.
\State For the remaining index $i$ sample two coins $x_i$-coins and restart unless their outcome is $0$ and $1$.
\State Output vertex $v$
\end{algorithmic}
\end{algorithm}
We can slightly optimize this procedure by sampling the coins first and then picking a random vertex that matches the coins:
\begin{algorithm}[H]
\caption{{\sc Bernoulli Factory for $\P_{\alpha,n}$ for non-integer $\alpha$} (version 2)}
\label{alg:sampford_generic_2}
\begin{algorithmic} 
\State Sample the $x_i$-coin for each $i \in [n]$. Let $S$ be the indices that returned $1$.
\State If $\abs{S} \neq k+1$, restart.
\State Choose $i \in S$ uniformly at random and flip the $x_i$-coin again. If the coin returns $1$, restart.
\State Output the vertex such that $v_i = \alpha-k$, $v_j = 1$ for $j \in S \setminus \{i\}$ and $v_j = 0$ for $j \notin S$.
\end{algorithmic}
\end{algorithm}

Now, if $\alpha$ is an integer, the polytope $\P_{\alpha,n}$ becomes non-generic. According to the recipe given in Theorem \ref{thm:non_generic} we perturb $k$ to $\alpha = k \pm \epsilon$ and take the limit as $\epsilon$ goes to zero. Depending on the sign of the perturbation this can lead to two different factories. 
One option is to look at the factory for $k-\epsilon$ and round vertices $(1-\epsilon, 1, \hdots, 1, 0, \hdots, 0)$ to $(1, 1, \hdots, 1, 0, \hdots, 0)$. Following \Cref{eq:nongeneric_poly} if we have a vertex $v$ where $A = \{i| v_i = 1\}$ and $B = \{i| v_i = 0\}$ then the associated Bernstein polynomial is:
\begin{equation}\label{eq:sampford1}
P_v(x) = \prod_{i \in A} x_i \cdot \prod_{i \in B} (1-x_i) \cdot \left( k - \sum_{i\in A} x_i\right)
\end{equation}

\noindent Then factory obtained recovers the Sampford sampling \citep{sampford1967sampling} procedure :

\begin{algorithm}[H]
\caption{{\sc Sampford sampling (minus $\epsilon$ version)}}
\label{alg:sampford_minus}
\begin{algorithmic} 
\State Sample the $x_i$-coin for each $i \in [n]$. Let $S$ be the indices that returned $1$.
\State If $\abs{S} \neq k$, restart.
\State Choose $i \in S$ uniformly at random and flip the $x_i$-coin again. If the coin returns $1$, restart.
\State Output the vertex such that $v_j = 1$ for $j \in S$ and $v_j = 0$ for $j \notin S$.
\end{algorithmic}
\end{algorithm}

\noindent A second option is to look at the factory for $k+\epsilon$ and round vertices $(\epsilon, 1, \hdots, 1,  0, \hdots, 0)$ to $(0, 1, 1, \hdots, 1, 0, \hdots, 0)$, which leads to the following polynomial:
\begin{equation}\label{eq:sampford2}
P_v(x) = \prod_{i \in A} x_i \cdot \prod_{i \in B} (1-x_i) \cdot \left( \sum_{i\in B} x_i\right)
\end{equation}
which is an alternative implementation of Sampford sampling.  The factory then becomes:

\begin{algorithm}[H]
\caption{{\sc Sampford sampling (plus $\epsilon$ version)}}
\label{alg:sampford_plus}
\begin{algorithmic} 
\State Sample the $x_i$-coin for each $i \in [n]$. Let $S$ be the indices that returned $1$.
\State If $\abs{S} \neq k+1$, restart.
\State Choose $i \in S$ uniformly at random and flip the $x_i$-coin again. If the coin returns $1$, restart.
\State Output the vertex such that $v_j = 1$ for $j \in S \setminus \{i\}$ and $v_j = 0$ for $j \notin S \cup \{i\}$.
\end{algorithmic}
\end{algorithm}

Note that polynomials \eqref{eq:sampford1} and \eqref{eq:sampford2} are different polynomials and hence lead to different algorithms, but evaluate the same within the polytope $\P_k$.

\noindent\textbf{Can we terminate on the vertices?}
One interesting observation is that for both polynomials \eqref{eq:sampford1} and \eqref{eq:sampford2}, it is the case that for any vertex $v' \in V$ we have $\sum_{v \in V} P_v(v') = 0$. This means that although this factory terminates a.s. in the interior of $\P_{k,n}$ it does not terminate on the vertices (one can also directly check that the algorithms described above never terminate in the case that $x$ is a vertex of $\P_{k, n}$). Hence these are Bernoulli factories but not strong Bernoulli factories. Can this be fixed? For the special cases of $k=1$ and $k=n-1$ it is possible to obtain a strong Bernoulli factory. In the case $k=1$ take $P_i(x) = x_i$ (where $P_i$ is the polynomial corresponding to the vertex $e_i$ with $1$ in the $i$-th coordinate). It is easy to verify that $\sum_i x_i (x-e_i) = 0$ since $\sum_i x_i = 1$ and $\sum x_i e_i = x$. Similarly for $k=n-1$ we can take $P_i(x) = 1-x_i$ (where now $P_i$ corresponds to the vertex $1-e_i$ with $1$ in each coordinate except $i$). It is natural to ask whether this can also be done for other values of $k$. In \Cref{sec:impossibility_race_over_bernstein}, we give a partial negative answer to this question for any integral value of $k$ such that $1<k<n-1$.

\vspace{-3mm}
\begin{acks}[Acknowledgments]

The authors would like to thank Shaddin Dughmi, Jan Vondrak, Jason Hartline, and Bobby Kleinberg for several helpful comments. 

\end{acks}

\begin{appendix}

\end{appendix}

\bibliographystyle{imsart-number} 
\bibliography{bernoulli}       
\newpage
\begin{supplement}
 \stitle{}
 
\section{Missing Proofs of \Cref{sec:matching_factory}}
\label{apx:matching}

\subsection{Proof of Lemma \ref{thm:zls}}

\begin{proof}[Proof of Lemma \ref{thm:zls}]
Let $J$ be the $n$-by-$n$ all-$1$ matrix. We evaluate $\det[A+J]$ by performing row and column operations, and showing $n^2\cof_{i,j}[A]=\det[A+J]$ -- which in turn proves the statement in the theorem. For simplicity, consider the first row and column of $A+J$ (same argument applies to any row/column). Add all other rows to the first one. Now all the entries of the first row are equal to $n$, while the rest of the matrix is unaffected. Then add all other columns to the first column. Now the $(1,1)$ entry is equal to $n^2$, and every other entry of the first row and the first column is equal to $n$. Factor out $n$ from the first column (so it becomes $[n,1,\ldots,1]^T$), and subtract the first column from every other column. Now the $(1,1)$ entry is $n$, every other entry of first row is zero, every other entry of first column is $1$, and more importantly every entry $(i,j)$ not in the first row or column is \emph{exactly} equal to $A_{i,j}$, as we subtracted $1$ from each such entry $(i,j)$ of $A+J$. By writing the determinant with respect to entry $(1,1)$, we have $\det[A+J]=n^2\cof_{1,1}[A]$, as desired. 
\end{proof}

\subsection{Proof of Lemma~\ref{lemma:bijection}}

We first show the following lemma, which proves to be useful for showing that every $r$-bi-tree can be uniquely decomposed into a matching and arborescence.

\begin{lemma}\label{lem:bitreedecomp}
Let $G$ be an $r$-bi-tree, and fix a $j \in [n]$. Then there exists a unique matching in $G$ of size $n-1$ where $j_R$ remains unmatched.
\end{lemma}
\begin{proof}
Note that since $G$ is an $r$-bi-tree, $j_R$ belongs to the unique non-empty connected component of $G$. \revcolor{This connected component is a tree on $2n-1$ vertices apart from $r_L$, simply because $G$ is bipartite and every vertex on the left has degree exactly $2$, which means this connected component has exactly $2n-2$ edges with $2n-1$ vertices --- hence it is a tree. } For any vertex $v \neq r_L$ in $G$, let its distance from $j_R$, \revcolor{denoted by $\textrm{dist}(j_R,v)$,}  be the length of the unique path from $j_R$ to $v$. 

Consider the set of all edges $E_{j}$ which connect a  node at distance $2k-1$ from $j_R$ with a node at distance $2k$ from $j_R$ \revcolor{for some integer $k$ },\revcolor{i.e., 
$$
E_j=\{e=(i_L,i'_R):e\in E(G),\exists k\in\mathbb{N}: \textrm{dist}(j_R,i_L)=2k-1,\textrm{dist}(j_R,i'_R)=2k\}
$$}
We first claim that $E_j$ is a matching in $G$ of size $n-1$ where $j_R$ remains unmatched. To see this, note that by the bipartite structure of $G$, each node $a_L \neq r_L$ is an odd distance $d$ from $j_R$. Exactly one of the two neighbors of $a_L$ must be at distance $d+1$ from $j_R$ (the other is at distance $d-1$). We therefore have $n-1$ edges, one incident to each $a_L \neq r_L$. To show that this is a valid matching, we must show that no two edges are incident to the same vertex $b_R$ on the right. Assume there is some vertex $b_R$ on the right matched with two vertices on the left, $a_L$ and $a'_L$. Then by the construction of $E_{j}$ if $b_R$ has distance $d$ from $j_R$, both $a_L$ and $a'_L$ have distance $d-1$ from $j_R$ and are connected to $b_R$. But this implies there are at least two different shortest paths from $j_R$ to $b_R$, contradicting the fact that the connected component is a tree; it follows that $E_{j}$ is a matching. Finally, since $j_R$ is at distance $0$ from itself, $E_j$ contains no edges matching $j_R$. 

It remains to show that this matching $E_j$ is the unique matching of size $n-1$ where $j_R$ remains unmatched. Assume to the contrary that there is another matching $E'_j$ of size $n-1$ in $G$ where $j_R$ remains unmatched. Since $E'_j \neq E_j$, $E'_j$ must contain an edge connecting a node at distance $2k$ from $j_R$ with a node at distance $2k+1$ from $j_R$. Let $d$ be the minimum value of $k$ such that there is an edge in $E_j$ connecting vertex $v_R$ at distance $2k$ with vertex $w_L$ at distance $2k+1$. If $d = 0$, then $v_R = j_R$, contradicting that $j_R$ is not matched. Otherwise, note that $v_R$ must be adjacent to some vertex $w'_L$ at distance $2d-1$ from $j_R$. Now, since $|E'_j| = n-1$, $w'_L$ must be matched (since there are only $n-1$ non-isolated nodes on the left), and $w'_L$ cannot be matched to $v_R$ (since $v_R$ is already matched to $w_L$). Since $w'_L$ is on the left, it has a unique other neighbor $v'_R$ distinct from $v_R$, and it follows that $w'_L$ must be matched to $v_R$. But $v'_R$ must have distance $2d-2$ from $j_R$, and we now have an edge connecting a node at distance $2d-2$ and a node at distance $2d-1$, contradicting the minimality of $d$.
\end{proof}

We are now ready to provide a detailed proof of \Cref{lemma:bijection}. 
\begin{proof}[Proof of Lemma~\ref{lemma:bijection}]
 To begin, we will show that $G(\pi, T)$ is always an $r$-bi-tree; i.e. that $G(\pi, T) \in \mathcal{G}_r$. To prove this, it suffices to check that $G(\pi, T)$ satisfies the three conditions in the definition of an $r$-bi-tree (\Cref{def:rbitree}):

\begin{enumerate}
    \item \textbf{The vertex $r_L$ is an isolated vertex:} The vertex $r_{L}$ is explicitly excluded from the edges from the matching $\pi$. The vertex $r$ is the root of the arborescence $T$ and thus has outdegree 0 and contributes no edges containing $r_L$ to the arborescence component of $G(\pi, T)$. 
    \item \textbf{The remainder of the vertices (aside from $r_{L}$) belong to a single connected component:} To see this, first add the edges from the matching. This creates $n-1$ connected components of size $2$, each containing a pair of vertices of the form $\{v_L, \pi(v)_R\}$. Contract all these components into single vertices, and identify each such component with its left vertex $v_L$; to distinguish it from the original vertex $v_L$, we will label this contracted vertex $\overline{v}$. For convenience of notation, we will additionally relabel the isolated vertex $\pi(r)_{R}$ as $\overline{r}$. Now, note that each edge $u \rightarrow v$ in the arborescence adds an edge from $u_L \in \overline{u}$ to $\pi(v)_R \in \overline{v}$. Since in the arborescence there is a path from any vertex to $r$, adding the edges from the arborescence to the bipartite graph implies there is a path from any component to $\overline{r}$, and therefore the vertices in this graph (aside from $r_L$) form a single connected component.
    \item \textbf{Each vertex $j_{L}$ (where $j \neq r$) on the left side has degree $2$:} The vertex $j_L \neq r_L$ is connected to one vertex $\pi(j)_R$ through the matching. Since $j$ is a non-root vertex in the arborescence, it has outdegree $1$ and there exists some edge $j \rightarrow p(j)$ in the arborescence. This contributes the edge $(j_L, \pi(p(j))_{R})$ to $G(\pi, T)$ (note that $\pi(p(j))_R \neq \pi(j)_R$ since $\pi$ is a permutation and $p(j) \neq j$). The vertex $j_L$ belongs to no other edges, and thus has degree $2$.
\end{enumerate}

To complete the bijection, we must show that for any $r$-bi-tree $G'$, there is a unique $(\pi, T)$ pair (with $\pi(r) = c$) such that $G' = G(\pi, T)$. Note that if $G' = G(\pi, T)$, then $\pi$ must correspond to a matching of size $n-1$ in $G'$ where all vertices are matched except $r_L$ and $\pi(r)_R$. If we further impose that $\pi(r) = c$, then $\pi$ must correspond to a matching of size $n-1$ in $G'$ where all vertices are matched except $r_L$ and $c_R$. By Lemma \ref{lem:bitreedecomp}, there is a unique such matching $\pi$ contained in $G'$. Removing the edges corresponding to this matching leaves $n-1$ edges remaining in $G'$. Along with the knowledge of $\pi$, this can be converted uniquely into a directed graph $T$ with $n-1$ edges: for each edge $(u_L, v_R)$ remaining in $G'$, there is a directed edge from $u$ to $\pi^{-1}(v)$ in $T$. 

It now remains to show that $T$ is an arborescence rooted at $r$. Since $T$ has $n-1$ edges, it suffices to show that from any vertex $v$ it is possible to reach $r$. To show this, we will show there is a sequence of vertices $v = v^{(1)}, v^{(2)}, \dots, v^{(k)} = r$ in $T$ such that there exists a path of the form 

$$v^{(1)}_L \rightarrow \pi(v^{(2)})_R \rightarrow v^{(2)}_L \rightarrow \pi(v^{(3)})_R \rightarrow \dots \rightarrow \pi(v^{(k)})_R$$

\noindent
in $G'$. By the construction of $T$, this implies there exists a directed path $v^{(1)} \rightarrow v^{(2)} \rightarrow \dots \rightarrow v^{(k)}$ in $T$ and thus a path from $v$ to $r$. To see that such a path exists, call the edges in $G'$ belonging to the matching $\pi$ ``matching edges'' and the remaining edges ``arborescence edges''. Note that each vertex on the left (except for $r_L$) is incident to exactly one matching edge and exactly one arborescence edge; each vertex on the right (except for $c_R$) is incident to exactly one matching edge. Therefore, repeatedly execute the following procedure, starting from $v_L$: follow the arborescence edge out of $v_L$ to some vertex $w_R$, and follow the matching edge from $w_R$ back to some vertex $v'_L$. This procedure must either end up at $c_R$ at some point (in which case there is no matching edge out of $c_R$ so we terminate) or it must end up in a cycle. However, since the connected component containing $v_L$ in $G'$ is a tree, we cannot end up in a cycle -- it follows that such a path exists to $c_R$, and therefore that $T$ is an arborescence.

\end{proof}
\section{Algebraic Proof of \Cref{prop:matching-factory} (by Darij Grinberg)}\label{appendix:algebraic_proof}

The combinatorial proof of  \Cref{prop:matching-factory}  constructs a bijection between the terms of the polynomial in \Cref{eq:rootatr} and the set of r-bi-trees. Here we present an alternative proof of \Cref{eq:rootatr} communicated to us by Darij Grinberg. The alternate proof is very elegant and purely algebraic.

To recap, the main goal in  \Cref{prop:matching-factory} is to show that for the polynomials $$Q_\pi(x) = \sum_{T\in\mathcal{T}_r(n)}\prod_{(u, v) \in T} x_{u,\pi(u)}x_{u, \pi(v)}$$
the sum $\sum_{\pi \in S_n |\pi(r)=c} Q_\pi(x)$ is independent of $c$. For that, the proof will construct a symmetric polynomial $S(x)$ that depends on neither $c$ nor $r$. Then we obtain the sum above as a projection depending only on $r$ in the ring of polynomials, therefore concluding that the expression is independent of $c$.

\paragraph*{Projection to polynomials of homogeneous multi-degree} Consider the ring $$R = \Z[x_{ij}; 1 \leq i,j \leq n]$$ and associate with each monomial $\prod_{ij} x_{ij}^{a_{ij}}$ the multi-degree $(\sum_j a_{1j}, \sum_j a_{2j}, \hdots, \sum_j a_{nj})$. Now, let $2e_{-r}$ be the multi-degree-vector that has $0$ in the $r$-th component and $2$ in any other component and $F_r : R \rightarrow R$ be the $\Z$-linear map that sends each polynomial in $R$ to its homogeneous component of multi-degree  $2e_{-r}$. In other words, it annihilates all but the monomials with multi-degree $2e_{-r}$.

\paragraph*{Symmetric polynomial} First define terms $y_{uv} \in R$ that are symmetric in $u$ and $v$ ($y_{uv} = y_{vu}$) as follows:
$$y_{uv} = \sum_{i=1}^n x_{iu} x_{iv}$$
and then use them to define the polynomial:
$$S(x) = \sum_{T \in \mathcal{T}(n)} \prod_{(u,v) \in T} y_{uv}  = \sum_{T \in \mathcal{T}_c(n)} \prod_{(u,v) \in T} y_{uv} $$
where $\mathcal{T}(n)$ (without a subscript) represents the set of trees on the set $[n]$. This is well defined since the terms $y_{uv}$ are symmetric.  Since the term is symmetric, we can root it at any point and replace the trees in $\mathcal{T}(n)$ by the arborescences $\mathcal{T}_c(n)$ (as done in the second equality above).

For each arborescence $T \in \mathcal{T}_c(n)$ we can re-write the term $\prod_{(u,v) \in T} y_{uv}$ as follows:

$$\prod_{(u,v) \in T} y_{uv} = \prod_{(u,v) \in T} \sum_i x_{iu} x_{iv} = \sum_{\alpha:T \rightarrow[n]} \prod_{(u,v) \in T} x_{\alpha(u,v), u} x_{\alpha(u,v), v} 
$$
where the sum in the last term is a sum over all maps $\alpha:T \rightarrow [n]$. The last equality follows simply from expanding the product. Using the fact that each arborescence $T \in \mathcal{T}_c(n)$ has exactly one outgoing edge from each vertex $u \in [n] \setminus c$ and no outgoing edge from $c$ we can re-write the above as:
$$\prod_{(u,v) \in T} y_{uv} =  \sum_{\sigma:[n] \setminus c \rightarrow [n] } \prod_{(u,v) \in T} x_{\sigma(u), u} x_{\sigma(u), v} =  \sum_{\substack{\sigma:[n] \rightarrow [n]\\ \sigma(c)=r} } \prod_{(u,v) \in T} x_{\sigma(u), u} x_{\sigma(u), v}
$$
where the second equality follows from extending the maps $\sigma:[n] \setminus c \rightarrow [n]$ to a map $\sigma:[n] \rightarrow [n]$ by mapping $c$ to $r$. Note that this doesn't alter the set of maps and it doesn't affect the expression since for all $(u,v) \in T$ we have $u \neq c$ (since the arborescence is rooted at $c$).

Putting it all together we can re-write:
$$S(x) =  \sum_{\substack{\sigma:[n] \rightarrow [n]\\ \sigma(c)=r} }  \sum_{T \in \mathcal{T}_c(n)} \prod_{(u,v) \in T} x_{\sigma(u), u} x_{\sigma(u), v}$$
Note that although $c$ and $r$ are in the formula above, the polynomial $S(x)$ itself doesn't depend on them since it is defined symmetrically.

\paragraph*{Projecting the polynomial} Now we apply the projection $F_r$ to the symmetric polynomial $S(x)$. Its effect is to annihilate a monomial $\prod_{(u,v) \in T} x_{\sigma(u), u} x_{\sigma(u), v}$ unless it has multidegree $2e_{-r}$. Note that multi-degree of this monomial is $\sum_{(u,v) \in T} 2e_{\sigma(u)} = \sum_{u \in [n]\setminus c} 2e_{\sigma(u)}$ which is $2e_{-r}$ if and only if $\sigma$ is a permutation. This allows us to write the projected polynomial as:
$$F_r(S)(x) =  \sum_{\substack{\sigma \in S_n\\ \sigma(c)=r} } \sum_{T \in \mathcal{T}_c(n)} \prod_{(u,v) \in T} x_{\sigma(u), u} x_{\sigma(u), v} $$
where the only modification is to sum over all permutations $\sigma \in S_n$ instead of all maps $\sigma:[n] \rightarrow [n]$.  Finally we re-write the sum in terms of the inverse permutation $\pi = \sigma^{-1} \in S_n$ and use the fact that iterating over all trees $\pi^{-1}(T) = \{(\pi^{-1}(u), \pi^{-1}(v)); (u,v) \in T\}$ is the same as iterating over all trees in $\mathcal{T}_r(n)$. Doing that we obtain:
$$\begin{aligned}
F_r(S)(x) & =  \sum_{\substack{\pi \in S_n\\ \pi(r)=c} } \sum_{T \in \mathcal{T}_c(n)} \prod_{(u,v) \in T} x_{\pi^{-1}(u), u} x_{\pi^{-1}(u), v} \\
& = \sum_{\substack{\pi \in S_n\\ \pi(r)=c} } \sum_{T \in \mathcal{T}_c(n)} \prod_{(u,v) \in \pi^{-1}(T)} x_{u, \pi(u)} x_{u, \pi(v)} \\
& = \sum_{\substack{\pi \in S_n\\ \pi(r)=c} } \sum_{T \in \mathcal{T}_r(n)} \prod_{(u,v) \in T} x_{u, \pi(u)} x_{u, \pi(v)}\\
& = \sum_{\substack{\pi \in S_n\\ \pi(r)=c} } Q_\pi(x)
\end{aligned}$$

\paragraph*{Conclusion} The statement that the sum $\sum_{\substack{\pi \in S_n; \pi(r)=c} } Q_\pi(x)$ doesn't depend on the choice of $c$ follows from the fact that this is equal to $F_r(S)$ and neither $S$ nor $F_r$ depend on the choice of $c$.

\begin{proof}[Proof of Lemma \ref{lem:bernstein_bound}]
Note that since $P(x)$ is a Bernstein polynomial, $P(x) \geq 0$ for $x \in [0, 1]^n$ (so $|P(x)| \leq \eps$ holds for all $x \in \P$). 

We need to restrict ourself to a subspace where $\P$ is full-dimensional in order to apply Lemma \ref{lem:wilhelmsen}. Assume $\H_{0}(\P)$ is $m$-dimensional for some $m \leq n$. Let $H$ be an orthogonal linear transformation mapping $\H_{0}(\P)$ to $\R^m$. Fix an arbitrary point $p_{0} \in \P$, and let $P_{\perp}(x): \R^m \rightarrow \R$ be the degree $m$ polynomial defined via $P_{\perp}(y) = P(p_{0} + H^{-1}y)$. Note that this same mapping maps $\P$ to a (full-dimensional) polytope $\P_{\perp} \subset \R^m$, so in particular $|P_{\perp}(y)| \leq \eps$ for all $y \in \P_{\perp}$. 

By Wilhelmsen's inequality (Lemma \ref{lem:wilhelmsen}), this implies that $|\partial_{w}P_{\perp}(y)| \leq (2d^2 / \omega(\P_{\perp}))\eps$ for any unit norm $w \in \R^m$. However, since $H$ is orthogonal, it is straightforward to verify that for $u \in \H_{0}(P)$ with $\Vert u \Vert=1$, that $\partial_{Hu}P_{\perp}(y) = \partial{u}P(p_{0} + H^{-1}y)$. Likewise, $\omega(\P_{\perp})$ is simply $\omega_{\H_{0}(\P)}(\P)$. The theorem statement follows.
\end{proof}

\section{Proofs of Theorems \ref{thm:necessary} and \ref{thm:necessary_strong}}
\label{app:necessary}

In this section we provide proofs of Theorem \ref{thm:necessary} and \ref{thm:necessary_strong}. We will prove Theorem \ref{thm:necessary} in two parts. We will first prove a necessary condition on the face structure of a polytope $\P$ for which it is possible to construct a Bernoulli factory. We will then show that this condition only holds for polytopes formed by the intersection of $[0, 1]^n$ and an affine subspace.

\paragraph*{Polyhedral Combinatorics}

We begin with some preliminaries from polyhedral combinatorics. Given a polytope $\P \subset \R^n$, we say a subset $F \subseteq \P$ is a  \textit{face} of $\P$ if $F = \arg\max_{p \in \P} w^{T}p$ for some vector $w \in \R^n$; in other words, $F$ is the set of points maximizing a linear functional over $\P$. The dimension of a face $F$ is the smallest dimension of an affine subspace of $\R^n$ containing $F$. In three dimensions, for example, the vertices of $\P$ are its $0$-dimensional faces, the edges of $\P$ are its $1$-dimensional faces, the facets of $\P$ are its $2$-dimensional faces, and $\P$ itself is its own $3$-dimensional face (assuming $\P$ is full-dimensional). 

The faces of a polytope $\P$ form a graded lattice under containment \citep{ziegler2012lectures}. Given a face $F$, we define the corresponding \textit{open face} $\tilde{F}$ to be the set of points in $F$ which belong to no lower-dimensional faces. Note that the open faces of $\P$ partition $\P$, since each point in $\P$ belongs to a unique maximal face. Let $D_{\P}(x)$ be the set of vectors $w$ such that $x \in \arg \max_{p \in \P} \dot{w}{p}$. The following alternate characterization of open faces will prove useful.

\begin{lemma}\label{lem:open_equiv}
Two points $x, x' \in \P$ belong to the same open face of $\P$ iff $D_{\P}(x) = D_{\P}(x')$. 
\end{lemma}
\begin{proof}
First, assume $D_{\P}(x) \neq D_{\P}(x')$. We will then show that $x$ and $x'$ cannot belong to the same open face of $\P$. If $D_{\P}(x) \neq D_{\P}(x')$, then without loss of generality, there exists a $w \in D_{\P}(x)$ such that $w \not\in D_{\P}(x')$. This means that $x$ belongs to the face $\arg\max_{p\in \P}\dot{w}{p}$, but $x'$ does not belong to this face. Since there is a face that $x$ belongs to but not $x'$, $x$ and $x'$ cannot belong to the same open face of $\P$. 

Now, assume that $x$ and $x'$ belong to different open faces of $\P$. We will show that $D_{\P}(x) \neq D_{\P}(x')$. Since $x$ and $x'$ belong to different open faces, there must be a face that one point belongs to that the other does not; without loss of generality, $x$ belongs to some face $F$ that $x'$ does not belong to. This face $F$ is equal to $\arg \max_{p \in \P} \dot{w}{p}$ for some $w$; it follows that $w \in D_{\P}(x)$ but $w \not\in D_{\P}(x')$ so $D_{\P}(x) \neq D_{\P}(x')$.
\end{proof}

If a point $x$ belongs to an open face of $\P$, this implies constraints on representing $x$ as a convex combination of other points in $\P$. 

\begin{lemma}\label{lem:convex}
Let $x$ be a point belonging to the open face $\tilde{F}$ of $\P$. Let $y_1, y_2, \dots, y_m \in \P$ be $m$ other points in $\P$ such that $x = \sum_{i=1}^{m}\lambda_iy_i$ for some coefficients $\lambda_i > 0$ satisfying $\sum_{i=1}^{m}\lambda_i = 1$. Then:

\begin{enumerate}
    \item For each $i$, $y_i \in F$.
    \item For any face $G$ strictly contained in $F$, there exists an $i$ such that $y_i \not\in G$. 
\end{enumerate}
\end{lemma}
\begin{proof}
To show 1, note that since $x \in \tilde{F} \subset F$, there exists some vector $w$ such that $x \in \arg\max_{p \in \P} \dot{w}{p}$. Since we can write $\dot{w}{x} = \sum_{i=1}^{n}\lambda_i \dot{w}{y_i}$, and since each $y_i \in \P$, it follows that each of the $y_i$ must also belong to $\arg\max_{p \in \P}\dot{w}{y_i}$ (and thus to $F$). 

To show 2, note that if all $y_i$ belong to $G$, then $x$ belongs to $G$ (since $G$ is convex and $x$ is a convex combination of the $y_i$). But since $x$ belongs to the open face $\tilde{F}$, $x$ cannot belong to any face $G$ strictly contained in $F$. 
\end{proof}

Next, let us consider two nested polytopes $\P$ and $\cQ$ such that $\P \subseteq \cQ$. We claim that every open face of $\P$ belongs to a single open face of $\cQ$. 

\begin{lemma}\label{lem:containment}
Let $\P$ and $\cQ$ be polytopes in $\R^n$ with $\P \subseteq \cQ$. Let $\tilde{F}$ be an open face of $\P$. Then there exists an open face $\tilde{G}$ of $\cQ$ such that $\tilde{F} \subseteq \tilde{G}$.
\end{lemma}
\begin{proof}
Assume to the contrary that there exists an open face $\tilde{F}$ of $\P$ that is not contained entirely in an open face of $\cQ$. In particular, choose two points $x_1, x_2 \in \P$ such that $x_1$ belongs to the open face $\tilde{G}_1$ of $\cQ$ and $x_2$ belongs to the (distinct) open face $\tilde{G}_2$ of $\cQ$. Consider $D_{\cQ}(x_1)$ and $D_{\cQ}(x_2)$; since $x_1$ and $x_2$ belong to different open faces of $\cQ$, these sets differ by Lemma \ref{lem:open_equiv}. Without loss of generality, let $w$ belong to $D_{\cQ}(x_1)$ but not to $D_{\cQ}(x_2)$.

Since $p \in \cQ$ recall that this implies that $\dot{w}{x_1} = \max_{p \in \cQ} \dot{w}{p}$. Since $\P \subseteq \cQ$, this means $\dot{w}{x_1} = \max_{p \in \P} \dot{w}{p}$, and therefore $w \in D_{\P}(x_1)$. Since $x_1$ and $x_2$ belong to the same open face of $\P$, this means $w \in D_{\P}(x_2)$. Finally, this implies that $\dot{w}{x_2} = \max_{p \in \P} \dot{w}{p} = \dot{w}{x_1}$ -- but in this case, we also have that $\dot{w}{x_2} = \max_{p \in \cQ} \dot{w}{p} = \dot{w}{x_1}$, and that $w \in D_{\cQ}(x_2)$, contradicting our earlier assumption. It follows that $D_{\cQ}(x_1) = D_{\cQ}(x_2)$ and that $\tilde{F}$ is contained within a single open face $\tilde{G}$ of $\cQ$. 
\end{proof}

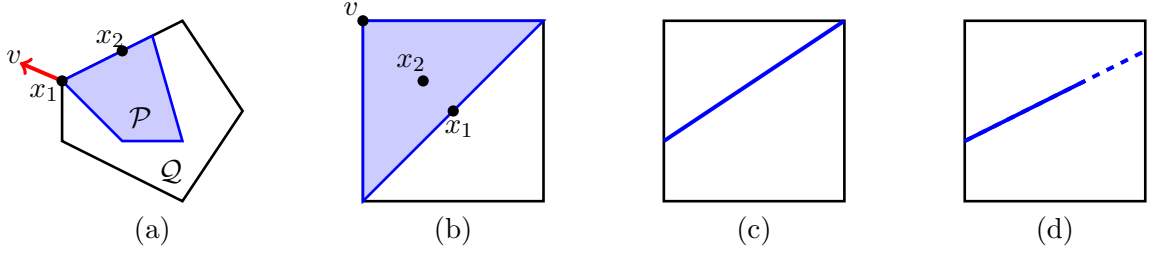
\begin{figure}[h]
\centering
\begin{tikzpicture}[scale=.8]
  \draw[line width=1pt] (0,0)--(0,1)--(2,2)--(3,.5)--(2,-1)--cycle;
  \draw[->,line width=1.5pt, color = red] (0,1)--(-.7,1.3);
  \node at (-.8,1.4) {$v$};
  \draw[line width=1pt, color=blue, fill=blue!20!white] (0,1)--(1.5,1.75)--(2,0)--(1,0)--cycle;
  \node[circle,fill,inner sep=1.5pt] at (0,1) {};
  \node[circle,fill,inner sep=1.5pt] at (1,1.5) {};
  \node at (-.3,.8) {$x_1$};
  \node at (.8,1.7) {$x_2$};
  \node at (1.3,.4) {$\P$};
  \node at (1.8,-.5) {$\cQ$};
  \node at (1.5,-1.5) {(a)};
  
\begin{scope}[xshift=5cm]
\draw[line width=1pt] (0,-1)--(0,2)--(3,2)--(3,-1)--cycle;
\draw[line width=1pt, color=blue, fill=blue!20!white] (0,-1)--(0,2)--(3,2)--cycle;
  \node at (1.5,-1.5) {(b)};
    \node[circle,fill,inner sep=1.5pt] at (1.5,.5) {};
  \node[circle,fill,inner sep=1.5pt] at (1,1) {};
  \node[circle,fill,inner sep=1.5pt] at (0,2) {};
  \node at (1.6,.2) {$x_1$};
  \node at (.8,1.3) {$x_2$};
  \node at (-.2,2.2) {$v$};
\end{scope}

\begin{scope}[xshift=10cm]
\draw[line width=1pt] (0,-1)--(0,2)--(3,2)--(3,-1)--cycle;
\draw[line width=1.5pt, color=blue, fill=blue!20!white] (0,0)--(3,2);
  \node at (1.5,-1.5) {(c)};
\end{scope}

\begin{scope}[xshift=15cm]

\draw[line width=1pt] (0,-1)--(0,2)--(3,2)--(3,-1)--cycle;
\draw[line width=1.5pt, color=blue, dashed] (0,0)--(3,1.5);
\draw[line width=1.5pt, color=blue, fill=blue!20!white] (0,0)--(2,1);
  \node at (1.5,-1.5) {(d)};
\end{scope}

\end{tikzpicture}
\caption{Figure (a) illustrates the proof of Lemma \ref{lem:containment}. Figure (b) illustrates the proof of Lemma \ref{lem:char1}. If $\P$ is the blue polytope it is impossible to build a Bernoulli factory, since at $x_2$ the factory should output vertex $v$ with non-zero probability and at $x_1$ the factory should never output $v$. It is impossible for a Bernoulli factory to put zero probability on the event of outputting $v$ at $x_1$ and non-zero probability at $x_2$. Figure (c) is an example where every open face of $\P$ (blue polytope) is contained in a different open face of $\cQ$. Finally, (d) is an illustration of the proof of Lemma \ref{lem:linear}. The solid line corresponds to $\P$ and the dashed line is the extension to $\cQ$.}
\label{fig:necessary_conditions_appendix}
\end{figure}

\paragraph*{Step 1: Faces in the interior}

Lemma \ref{lem:containment} is important for us since it implies that for any polytope $\P \subseteq [0, 1]^n$ that the open faces of $\P$ are contained in open faces of $[0, 1]^n$. Of the open faces of $[0, 1]^n$, we especially care about the $n$-dimensional interior $(0, 1)^n$, since this contains the domain any (non-strong) Bernoulli factory for $\P$. 

We first show that a discrete factory which has positive probability of outputting an element $s$ somewhere in $(0, 1)^n$ has a positive probability of outputting $s$ \textit{everywhere} in $(0, 1)^n$. 

\begin{lemma}\label{lem:bfpos}
Let $\mathcal{F}(x)$ be a discrete factory to a finite set $S$. Then if $\Pr[\mathcal{F}(x) = s] > 0$ for any $x \in (0, 1)^n$, $\Pr[\mathcal{F}(x) = s] > 0$ for every $x \in (0, 1)^n$. 
\end{lemma}
\begin{proof}
Assume that $\Pr[\mathcal{F}(x) = s] > 0$ for a fixed $x \in (0, 1)^n$. This means that there exists a leaf $\ell$ in the protocol tree for $\mathcal{F}$ labelled with $s$ such that $\Pr[\mathcal{F}(x) \rightarrow \ell] > 0$. This probability $\Pr[\mathcal{F}(x) \rightarrow \ell]$ can also be written as some (scaled) Bernstein monomial $\pi(x) = c\prod_{i} x_i^{a_i}(1-x_i)^{b_i}$ (where $c > 0$ since $\pi(x) > 0$). But then $\pi(x) > 0$ for all $x \in (0, 1)^n$, and therefore $\Pr[\mathcal{F}(x) = s] > 0$.
\end{proof}

We are now ready to prove the first step of our argument: that if $\P$ contains two open faces that belong to the interior of $[0, 1]^n$, then there does not exist a Bernoulli factory for $\P$. 

\begin{lemma}\label{lem:char1}
Let $\tilde{F}_1$ and $\tilde{F}_2$ be two different open faces of a polytope $\P \subseteq [0, 1]^n$. If $\tilde{F}_1$ and $\tilde{F}_2$ are both contained in $(0, 1)^n$, then it is impossible to build a Bernoulli factory for $\P$. 
\end{lemma}
\begin{proof}
Assume to the contrary that there exists a Bernoulli factory for such a $\P$. Choose a point $x_1 \in \tilde{F}_1$ and a point $x_2 \in \tilde{F}_2$; by assumption, both $x_1$ and $x_2$ also belong to $(0, 1)^n$. 

If we run our Bernoulli factory on a point $x \in \P$, it will output each of the vertices of $\P$ with some probability. Let $V(x)$ be the subset of these vertices which are output with positive probability. We first claim that since $x_1$ and $x_2$ belong to different open faces of $\P$, $V(x_1) \neq V(x_2)$. To see this, let $G = F_1 \cap F_2$; since $F_1 \neq F_2$, $G$ is strictly contained within at least one of $F_1$ or $F_2$; without loss of generality $G \subset F_1$. Now, note that (by the definition of $V(x)$), it is possible to write $x$ as a positive convex combination of the vertices in $V(x)$. By condition (2) of Lemma \ref{lem:convex}, this means there exists a vertex $v \in V(x_1)$ such that $v \not\in G$ (and thus, $v \not\in F_2$). But by condition (1) of Lemma \ref{lem:convex}, this means that every vertex $v' \in V(x_2)$ satisfies $v' \in F_2$. It follows that $V(x_1) \neq V(x_2)$.

Now, without loss of generality, assume there exists a vertex $v$ which belongs to $V(x_1)$ but not to $V(x_2)$. Since $v \in V(x_1)$, this means that the Bernoulli factory has a positive probability of outputting $v$ on input $x_1$.  From Lemma \ref{lem:bfpos}, since $x_1$ and $x_2$ both lie in $(0, 1)^n$, this implies that the Bernoulli factory has a positive probability of outputting $v$ on input $x_2$. But this implies that $v \in V(x_2)$, contradicting our choice of $v$. It follows that no Bernoulli factory for $\P$ can exist, as desired.
\end{proof}

\paragraph*{Step 2: Affine intersections}

We now show that polytopes that don't satisfy the condition in Lemma \ref{lem:char1} are exactly the polytopes that can be written as the intersection of $[0, 1]^n$ and an affine subspace.

\begin{lemma}\label{lem:linear}
Let $\P \subseteq [0, 1]^n$ be a polytope such that $\P \cap (0, 1)^n \neq \emptyset$. If the interior of $\P$ is the unique open face of $\P$ contained in $(0, 1)^n$, then $\P$ is the intersection of $[0, 1]^n$ and an affine subspace.
\end{lemma}
\begin{proof}
Assume to the contrary that $\P$ is not the intersection of $[0, 1]^n$ and an affine space. We will show that there are two open faces of $\P$ that lie in the same open face of $[0, 1]^n$. 

Let $H$ be the affine span of $\P$ (i.e., the smallest affine subspace containing $\P$). Let $\cQ = [0, 1]^n \cap H$. By assumption, $\P$ is strictly contained in $\cQ$. In particular, this means that there exists a point $x$ on the boundary of $\P$ that lies in the interior of $\cQ$. The open face $\tilde{F}$ of $\P$ containing $x$ must also lie in the interior of $\cQ$. Since $\cQ \subseteq [0, 1]^n$, this means $\tilde{F}$ lies in $(0, 1)^n$. But the interior of $\P$ must also lie in $(0, 1)^n$, and is distinct from $\tilde{F}$ (since $x$ is on the boundary of $\P$). We thus have two open faces of $\P$ (a $\tilde{F}$ and $\P$'s interior) which both lie in $(0, 1)^n$. This contradicts our assumption and therefore $\P$ must be the intersection of $[0, 1]^n$ and an affine subspace.
\end{proof}

The proof of Theorem \ref{thm:necessary} now follows immediately.

\begin{proof}[Proof of Theorem \ref{thm:necessary}]
Follows from Lemmas \ref{lem:char1} and \ref{lem:linear}.
\end{proof}

\paragraph*{Necessary condition for strong Bernoulli factories}

Finally, we prove Theorem \ref{thm:necessary_strong}, the necessary condition for strong Bernoulli factories. We will be able to do this by reducing to Theorem \ref{thm:necessary}.

\begin{proof}[Proof of Theorem \ref{thm:necessary_strong}]
Note that any strong Bernoulli factory for $\P$ is also a regular Bernoulli factory for $\P$. Thus, if $\P \cap (0, 1)^n \neq \emptyset$, this impossibility follows from Theorem \ref{thm:necessary}.

Assume then that $\P \cap (0, 1)^n = \emptyset$. Then $\P$ is contained in some minimal face $F$ of $[0, 1]^n$. If $F$ is $m$-dimensional, it is isomorphic to $[0, 1]^m$ (in particular, we can think of $F$ as the set of points where we fix $n-m$ of the coordinates of $x$ and the remaining coordinates can range from $0$ to $1$). Let $\P'$ be the projection of $\P$ to $[0, 1]^m$ (we can think of projection here as simply omitting the fixed coordinates of $F$). Note that since $F$ is minimal, $\P' \cap [0, 1]^{m} \neq \emptyset$.

We claim that any strong Bernoulli factory $\mathcal{F}$ for $\P$ gives a regular Bernoulli factory for $\P'$; in particular, given a point $x' \in \P'$, we can transform it to a point $x \in \P$ by reintroducing the fixed coordinates, and run $\mathcal{F}(x)$. It follows from Theorem \ref{thm:necessary} that $\P'$ must be the intersection of an affine subspace with $[0, 1]^m$. But then $\P$ must be the intersection of an affine subspace with $[0, 1]^n$, as desired.

\end{proof}


\section{Missing Proofs of \Cref{sec:generic}}
\label{appendix:missingproofs}
\subsection{Proof of Lemma~\ref{lemma:disjoint_pm}}
\begin{proof}[Proof of Lemma~\ref{lemma:disjoint_pm}]
Consider a point $x$ in $Z_B^i$. This point can be written in the form:
$$x = w^i + \sum_{r \in I \setminus \{i\}} \lambda_r w^r \quad \text{for} \quad \lambda_r \in (0,1)$$
Now consider a point $x'$ in $Z_A^t$, which can similarly be written in the form:
$$x' = \sum_{r \in I \setminus \{t\}} \lambda'_r w^r \quad \text{for} \quad \lambda'_r \in (0,1)$$
Assume to the contrary that $x = x'$. We then have that

\begin{equation}\label{eq:decomp}
    (1 - \lambda'_i)w^i + \lambda_t w^t = \sum_{r \in I \setminus \{i, t\}} (\lambda'_r - \lambda_r)w^r.
\end{equation}

Note that the term on the right hand side of \eqref{eq:decomp} belongs to the hyperplane $\H$ spanned by the $(k-1)$ vectors $w^r$ where $r \not\in \{i, t\}$. We will now write:
\begin{equation}\label{eq:hperp}
 w^i =  h_i + c_ih_{\perp} \quad \text{and} \quad  w^{t} = h_t + c_th_{\perp}
\end{equation}
for $h_i, h_j \in \H$ and $h_{\perp}$  a unit vector orthogonal to $\H$. 

Observe now that $c_i$ and $c_t$ must have the same sign. First consider the case where neither is $j$. Then:
$$-1 = \frac{\sigma_i}{\sigma_t} = \sign \left( \frac{\det[W' w^j w^t]}{\det[W' w^i w^j]} \right) = - \sign \left( \frac{\det[W' w^j w^t]}{\det[W' w^j w^i]} \right)$$
where $W'$ is a matrix containing all columns $w^r$ except $w^i$, $w^j$, and $w^t$. Finally note that:
$$\det[W' w^j w^t] = c_t \det[W' w^j h_\perp], \quad \det[W' w^j w^i] = c_i \det[W' w^j h_\perp] $$
which shows that $\sign(c_i / c_j) = +1$. The case where one of the indices $\{i,t\}$ equals $j$ is analogous. If $t$ is $j$ then we have that:
$$-1 = \frac{\sigma_i}{\sigma_j} = \sigma_i = - \sign \left( \frac{\det[W' w^i ]}{\det[W' w^j]} \right)$$
where now the columns of $W'$ are formed by all the other $w^r$ except $w^i, w^j$. We again reach the same conclusion that $\sign(c_i/c_j) = +1$.   

Now, assume without loss of generality that $c_i$ and $c_j$ are both positive. But then, since $(1-\lambda'_i) > 0$ and $\lambda_t > 0$, the left hand side of (\ref{eq:decomp}) will have a positive $h_{\perp}$ component and cannot lie entirely in $H$. Thus it is not possible that $x = x'$, as desired. 
\end{proof}

\subsection{Proof of Lemma~\ref{lemma:disjoint_pp}}
\begin{proof}[Proof of Lemma~\ref{lemma:disjoint_pp}]
The proof follows a similar pattern as the proof of Lemma \ref{lemma:disjoint_pp}. We first choose points $x = \sum_{r \in I \setminus i} \lambda_r w^r \in Z_A^i$ and $x'=  \sum_{r \in I \setminus t} \lambda_r w^r\in Z_A^t$. Now, if $x=x'$ then $\lambda_i w^i - \lambda_t w^t$ should be contained in the hyperplane $\H$ spanned by the vectors $w^r$ for $r \neq i,t$. This again allows us to write $w^i$ and $w^t$ as in \Cref{eq:hperp}, but this time with $\sign(c_i / c_t) = -1$ since $\sigma_i / \sigma_t = +1$. With this sign pattern it is impossible to have $\lambda_i w^i - \lambda_t w^t$ in $\H$ since it will have a non-zero $h_\perp$ component.

The argument for $Z_B^i$ and $Z_B^t$ is the same.
\end{proof}
\section{Impossibility of Extending $k$-subset Factories to the Boundary}\label{sec:impossibility_race_over_bernstein}

In \Cref{sec:non_generic}, we observed that the Bernoulli factories we designed for the $k$-out-of-$n$ subset polytope $\Pnk$ are not strong Bernoulli factories -- i.e., they do not extend to the boundary of $[0, 1]^n$. It is natural to ask whether there do exist strong Bernoulli factories for these polytopes. In this section, we show that there is no ``nice'' strong Bernoulli factory for $\Pnk$ for integral $k$ satisfying $1 < k < n-1$.

To define what we mean by ``nice'', we need to introduce some auxiliary notation. Previously, we have restricted our attention to Bernoulli factories that terminate almost surely on their domain. In this section, we will want to restrict our attention to factories that not only terminate a.s., but that terminate quickly. Let $T_{\mathcal{F}}(x)$ be the random variable equal to the depth of the leaf node on which $\mathcal{F}(x)$ terminates (with $T_{\mathcal{F}}(x) = \infty$ if the execution of $\mathcal{F}(x)$ never terminates). $T_{\mathcal{F}}(x)$ represents the total number of coins flipped by the factory $\mathcal{F}$ in a single execution. We say that $\mathcal{F}$ \textit{converges exponentially} on a domain $S \subseteq [0,1]^n$ if there exists a constant $c < 1$ such that

$$\Pr[T_{\mathcal{F}}(x) > d] \leq c^d$$

\noindent
for all positive integers $d$ and $x \in S$. This notion of exponential convergence appears throughout the Bernoulli factory literature (for example, \citep{nacu2005fast} refer to this as ``fast simulation''); most known explicit Bernoulli factories have the property of exponential convergence. 

We prove the following theorem.

\begin{theorem}\label{thm:impossibility}
Let $k$ be an integer satisfying $1<k<n-1$. There is no strong Bernoulli factory for $\Pnk$ that converges exponentially.
\end{theorem}

Note in particular that any strong Bernoulli race that terminates a.s. converges exponentially since there is a constant probability of success in each iteration  (in particular, all strong Bernoulli factories we have introduced thus far converge exponentially on $\P$). We thus have the following corollary.

\begin{corollary}\label{thm:impossibility_race}
Let $k$ be an integer satisfying $1<k<n-1$. There is no Bernoulli race over Bernstein polynomials that is a strong Bernoulli factory for $\P_{k,n}$.
\end{corollary}

We will actually prove Theorem \ref{thm:impossibility} for a slightly weaker version of ``niceness'' based on the differentiability of the functions $\Pr[\mathcal{F}(x) = v]$. For a generic factory $\mathcal{F}$ 
we define $P_{v}(x) \triangleq \Pr[\mathcal{F}(x) = v]$ and $P_{v, T}(x) \triangleq \Pr[\mathcal{F}(x) = v \wedge T_{\mathcal{F}}(x) \leq T]$. Note that each $P_{v, T}(x)$ is the sum of finitely many Bernstein monomials (corresponding to leaves of $\mathcal{F}$ at depth at most $T$), and therefore is a Bernstein polynomial. The function $P_v(x)$ here is a Bernstein \emph{series}, i.e. the limit of the Berstein polynomials $P_{v,T}(x)$.\\

Given a polytope $\P$, let $\H(\P)$ be the minimum affine subspace containing $\P$ (the ``affine span'' of $\P$). Let $\H_{0}(\P) = \{v - v' \mid v, v' \in \H(\P)\}$ be the translate of $\H(\P)$ passing through the origin. 

\begin{definition}
A Bernoulli factory $\mathcal{F}$ for a polytope $\P$ is \emph{differentiable} if for each $v \in V$ and each $u \in \H_{0}(\P)$ with $\Vert u \Vert = 1$, the derivative $\partial_{u}P_{v}(x)$ exists and is equal to the limit $\lim_{T\rightarrow\infty}\partial_{u} P_{v, T}(x)$.
\end{definition}

In other words, a Bernoulli factory $\mathcal{F}$ is differentiable if the function $P_{v}(x)$ is differentiable on the minimal subspace containing the polytope $\P$ and if these derivatives can be recovered as limits of the derivatives of the polynomials $P_{v, T}(x)$.

We will prove that there is no differentiable strong Bernoulli factory for $\P_{k,n}$ (Lemma \ref{lem:impossibility_uniform}) and then argue that all exponentially converging factories are differentiable (Theorem \ref{thm:exponential_implies_differentiable}). We begin by proving the following structural result on the polynomials $P_{v, T}(x)$ for any strong Bernoulli factory for $\P_{k,n}$.

\begin{lemma}\label{lemma:Pv_divisors}
Let $k$ be an integer satisfying $1 < k < n-1$. Let $\mathcal{F}$ be a strong Bernoulli factory for $\P_{k,n}$ and let $v$ be a vertex of $\P_{k,n}$. Then for any $T \geq 0$, the polynomial $P_{v, T}(x)$ must be divisible by $\prod_{i| v_i=1} x_i~\prod_{i| v_i=0} (1-x_i)$.
\end{lemma}
\begin{proof}
Since $v$ is a vertex of $\P_{k,n}$, $v$ has exactly $k$ coordinates equal to $1$ and $n-k$ coordinates equal to $0$. Assume without loss of generality that $v_1 = \dotsb = v_k = 1$ and $v_{k+1} = \dotsb = v_n = 0$. Let $\ell$ be a leaf in $\mathcal{F}$ with label $v$ and depth at most $T$. Note that $\Pr[\mathcal{F}(x) \rightarrow \ell]$ is a Bernstein monomial; let $M_{\ell}(x) = \Pr[\mathcal{F}(x) \rightarrow \ell]$. Assume $M_{\ell}(x)$ is a non-zero monomial. Then:

\begin{itemize}
\item For each $i \in \{1, 2, \dots, k\}$, $M_{\ell}(x)$ must be divisible by $x_i$. If not, then $M_{\ell}(x)$ would be strictly positive at the point $x$ where $x_i = 0$ and $x_j = k/(n-1)$ for all $j \neq i$. But the vertex $v$ cannot appear with positive weight in a convex combination resulting in $x$ (since $v_i > 0$ and $v'_i \geq 0$ for all other vertices $v'$). 

\item Similarly, for each $i \in \{k+1, \dots, n\}$, the polynomial $M_{\ell}(x)$ must be divisible by $(1-x_i)$. If not, then $M_{\ell}(x)$ would be strictly positive at the point $x$ where $x_i = 1$ and $x_j = (k-1)/(n-1)$ for all $j \neq i$. Again, the vertex $v$ cannot appear with positive weight in a convex combination resulting in $x$ (since $v_i < 1$ and $v'_i \leq 1$ for all other vertices $v'$) .
\end{itemize}

Since we can write $P_{v, T}(x)$ as the sum of a finite number of such monomials $M_{\ell}(x)$, it follows that the Bernstein polynomial $P_{v, T}(x)$ must be divisible by $x_1 x_2 \dots x_k (1-x_{k+1}) \dots (1-x_n)$.

\end{proof}

Note that Lemma \ref{lemma:Pv_divisors} doesn't hold for $k=1$ or $k=n-1$ since $k/(n-1)$ and $(k-1)/(n-1)$ need to be strictly between $0$ and $1$ for the proof to work. This is an important sanity check as for $\P_{1,n}$ and $\P_{n-1,n}$ it is indeed possible to construct strong Bernoulli factories.

Lemma \ref{lemma:Pv_divisors} allows us to conclude that the gradient of $P_{v, T}(x)$ vanishes on vertices $v' \neq v$.

\begin{lemma}\label{lem:zero_gradient}
Let $k$ and $n$ be integers satisfying $1 < k < n-1$. Let $\mathcal{F}$ be a strong Bernoulli factory for $\P_{k,n}$ and let $v$ and $v'$ be two distinct vertices of $\P_{k,n}$. Then for any $T \geq 0$, $\nabla P_{v, T}(v') = 0$.
\end{lemma}
\begin{proof}
It suffices to show $\partial_i P_{v, T}(v') = 0$ for each $i \in [n]$. Note that by Lemma \ref{lemma:Pv_divisors}, for any $T \geq 0$, we can write $P_{v, T}(x) = \Pi(x) R_{v, T}(x)$ where $\Pi(x) = \prod_{i| v_i=1} x_i \cdot \prod_{i|v_i=0} (1-x_i)$ and where $R_{v, T}(x)$ is a Bernstein polynomial. Let $v' \in V$ be a vertex $v' \neq v$. Then we claim that the partial derivative $\partial_i P_{v, T}(v') = 0$ for all $i \in [n]$. To see this, note that
\begin{eqnarray*}
\partial_i P_{v, T}(x) &=& \partial_i \Pi(x) \cdot R_v(x)  + \Pi(x) \cdot \partial_i R_v(x).
\end{eqnarray*}

\noindent
Since $v' \in \{0, 1\}^n$ and $v'$ differs from $v$ in two coordinates, $\Pi(x)$ has at least two terms that evaluate to zero and hence $\partial_i \Pi(x) = 0$. It follows that $\partial_i P_{v, T}(v') = 0$.
\end{proof}

For differentiable Bernoulli factories, Lemma \ref{lem:zero_gradient} implies that the derivatives of $P_{v}(x)$ at vertices $v' \neq v$ are zero (along vectors in $\H_{0}(\P_{k,n})$) .

\begin{corollary}\label{cor:zero_gradient}
Let $k$ and $n$ be integers satisfying $1 < k < n-1$. Let $\mathcal{F}$ be a differentiable strong Bernoulli factory for $\P_{k,n}$ and let $v$ and $v'$ be two distinct vertices of $\P_{k,n}$. Let $u$ be a unit vector belonging to $\H_{0}(\P)$. Then $\partial_{u} P_{v}(v') = 0$.
\end{corollary}
\begin{proof}
By Lemma \ref{lem:zero_gradient}, $\nabla P_{v, T}(v') = 0$ for all $T \geq 0$, and thus $\partial_{u} P_{v, T}(v') = \langle u, \nabla P_{v, T}(v')\rangle = 0$. Since $\mathcal{F}$ is differentiable, we know that $\partial_{u} P_{v}(v')$ exists and equals $\lim_{T\rightarrow\infty}\partial_{u} P_{v, T}(v') = 0$.
\end{proof}

We can now prove impossibility for differentiable factories.

\begin{lemma}\label{lem:impossibility_uniform}
Let $k$ and $n$ be integers satisfying $1 < k < n-1$. There is no differentiable strong Bernoulli factory for $\Pnk$.
\end{lemma}
\begin{proof}
Assume to the contrary that a differentiable strong Bernoulli factory $\mathcal{F}$ exists for $\Pnk$. Let $Q(x) = \sum_{v} P_{v}(x)(v-x)$. Since $\mathcal{F}$ is a strong Bernoulli factory for $\Pnk$, we know that $Q(x) = 0$ holds for all $x \in \P$. Fix an $i \in [n]$, and let $Q_{i}(x)$ be the $i$th component of $Q(x)$. 

Since $Q_i(x) = 0$ on all of $\P_{k,n}$ (and since $\P_{k,n}$ is $(n-1)$-dimensional), it follows that the directional derivative of $Q_{i}(x)$ along any non-zero vector $u$ in $\H_{0}(\P_{k,n})$ is zero for any $x \in \P_{k,n}$. That is, for any non-zero $u$ in $\H_{0}(\P_{k,n})$ and $x \in \P_{k,n}$,

\begin{equation}\label{eq:directional}
    \partial_{u} Q_i(x) = 0.
\end{equation}

Note that since $Q_{i}(x) = \sum_{v}P_{v}(x)(v_i - x_i)$, by the product rule we have that

\begin{equation}\label{eq:product}
\partial_{u} Q_{i}(x) = \sum_{v}\partial_{u} P_{v}(x) (v_i - x_i) - \left(\sum_{v} P_{v}(x)\right)u_i.
\end{equation}

Now, let us evaluate $\partial_{u} Q_{i}(v')$ for some vertex $v'$ of $\P_{k,n}$. Note that by Corollary \ref{cor:zero_gradient}, $\partial_{u} P_{v}(v') = 0$ for any vertex $v \neq v'$ of $\Pnk$; on the other hand, $\partial_{u} P_{v'}(x)(v'_i - x_i) = 0$ when $x = v'$ since then $v'_i - x_i = 0$. We also know that for any vertex $v \neq v'$, $P_{v}(v') = 0$ (since if $x = v'$, $\mathcal{F}(x)$ can only output $v'$). Therefore when $x = v'$, \Cref{eq:product} simplifies to

\begin{equation}\label{eq:product2}
\partial_{u} Q_{i}(v') = -P_{v'}(v')u_i,
\end{equation}

Substituting this into \Cref{eq:directional}, we have that

\begin{equation}\label{eq:directional2}
    -P_{v'}(v')u_i = 0.
\end{equation}

Now, recall that $\H(\P_{k,n}) = \{u \mid \sum_{i}u_{i} = k\}$, and thus $\H_{0}(\P_{k,n}) = \{u \mid \sum_{i}u_{i} = 0 \}$. Since $n \geq 3$ (since $n-1 > 1$) we can choose a vector $u \in \H_{0}(\P_{k,n})$ satisfying $u_i > 0$. This implies $P_{v'}(v') = 0$. However, if $P_{v'}(v') = 0$ then $\sum_{v} P_{v}(v') = 0$, contradicting the requirement for strong Bernoulli factories that $\sum_{v}P_{v}(x) = 1$ for all $x \in \P$. This implies that the assumed factory $\mathcal{F}$ cannot exist.
\end{proof}

Finally, we show that any factory that converges exponentially is differentiable, thus implying Theorem \ref{thm:impossibility}. To do so, we will need the following multivariate generalization of Markov brothers' inequality due to \citep{wilhelmsen1974markov}.

\begin{lemma}\label{lem:wilhelmsen}
Let $T$ be a compact, convex set in $\R^n$ with non-empty interior. Let $\omega(T)$ be the minimum width of $T$ in any direction $u \in \R^n$; i.e. 
$\omega(T) = \min_{\Vert u \Vert = 1}(\max_{p \in T} \dot{u}{p} - \min_{p \in T} \dot{u}{p}).$
Let $P: \R^n \rightarrow \R$ be a degree $d$ multivariate polynomial satisfying $|P(x)| \leq \eps$ for all $x \in T$. Then for all $x \in T$ it holds that $\Vert \nabla P(x) \Vert \leq 2\eps d^2 / \omega(T).$
\end{lemma}

Wilhelmsen's inequality immediately implies the following lemma bounding the derivative of a Bernstein polynomial on a polytope $\P$.

\begin{lemma}\label{lem:bernstein_bound}
Let $\P \subseteq [0, 1]^n$ and let $P: [0, 1]^n \rightarrow \R$ be a Bernstein polynomial of degree at most $d$ that satisfies $P(x) \leq \eps$ for all $x \in \P$. Then for each $u \in \H_0(\P)$ with $\Vert u \Vert = 1$,
$$|\partial_u P(x)| \leq \frac{2d^2\eps}{\omega_{\H_{0}(\P)}(\P)}$$
for all $x \in \P$, where $\omega_{\H_{0}(\P)}(\P)$ denotes the width of $\P$ in the directions contained within $\H_{0}(\P)$:
$$\omega_{\H_{0}(\P)}(\P) = \min_{\Vert u \Vert = 1, u \in \H_{0}(\P)}\left(\max_{p \in \P} \dot{u}{p} - \min_{p \in \P} \dot{u}{p}\right).$$

\end{lemma}
\begin{proof}
Note that since $P(x)$ is a Bernstein polynomial, $P(x) \geq 0$ for $x \in [0, 1]^n$ (so $|P(x)| \leq \eps$ holds for all $x \in \P$). 

We need to restrict ourself to a subspace where $\P$ is full-dimensional in order to apply Lemma \ref{lem:wilhelmsen}. Assume $\H_{0}(\P)$ is $m$-dimensional for some $m \leq n$. Let $H$ be an orthogonal linear transformation mapping $\H_{0}(\P)$ to $\R^m$. Fix an arbitrary point $p_{0} \in \P$, and let $P_{\perp}(x): \R^m \rightarrow \R$ be the degree $m$ polynomial defined via $P_{\perp}(y) = P(p_{0} + H^{-1}y)$. Note that this same mapping maps $\P$ to a (full-dimensional) polytope $\P_{\perp} \subset \R^m$, so in particular $|P_{\perp}(y)| \leq \eps$ for all $y \in \P_{\perp}$. 

By Wilhelmsen's inequality (Lemma \ref{lem:wilhelmsen}), this implies that $|\partial_{w}P_{\perp}(y)| \leq (2d^2 / \omega(\P_{\perp}))\eps$ for any unit norm $w \in \R^m$. However, since $H$ is orthogonal, it is straightforward to verify that for $u \in \H_{0}(P)$ with $\Vert u \Vert=1$, that $\partial_{Hu}P_{\perp}(y) = \partial{u}P(p_{0} + H^{-1}y)$. Likewise, $\omega(\P_{\perp})$ is simply $\omega_{\H_{0}(\P)}(\P)$. The theorem statement follows.
\end{proof}

\begin{theorem}\label{thm:exponential_implies_differentiable}
If $\mathcal{F}$ is a strong Bernoulli factory for a polytope $\P$ that converges exponentially, $\mathcal{F}$ is differentiable. 
\end{theorem}
\begin{proof}
We will use the following fact (see e.g. Theorem 6.2.10 of \citep{lebl2014basic}). Let $f_1(x), f_2(x), \dots$ be continuously differentiable functions from $\R^n$ to $\R$ that converge pointwise to a function $f(x)$ on some compact convex subset $S \subset \R^n$. Then if (for some $u \in \R^n$ with $\Vert u \Vert=1$) the sequence $\partial_uf_1, \partial_uf_2, \dots$ converges uniformly to a function $g$ over all $x \in S$, $\partial_u f$ exists and is equal to $g$ on $S$. 


It thus suffices to show for each $u \in \H_{0}(\P)$ that the sequence $\partial_uP_{v, T}(x)$ as $T \rightarrow \infty$ converges uniformly. To do this, for each $T \geq 1$, let $\Delta_{v, T}(x) = P_{v, T}(x) - P_{v, T-1}(x)$ (and let $\Delta_{v, 0}(x) = P_{v, 0}(x)$). We then wish to show that the sum $\sum_{t=0}^{\infty} \partial_u\Delta_{v, t}(x)$ converges uniformly. To do so, observe that $\Delta_{v, t}(x)$ is a Bernstein polynomial of degree at most $t$ (since it is the sum of monomials corresponding to leaves at depth at most $t$). We also know (from the definition of exponential convergence) that there exists a $c < 1$ such that $\Delta_{v, t}(x) \leq c^{t}$ for all $x \in \P$. From Lemma \ref{lem:bernstein_bound}, it then follows that

$$|\partial_{i}\Delta_{v, t}(x)| \leq \frac{2t^2 c^{t}}{\omega_{\mathcal{H}_0(\P)}(\P)}.$$

Since the sum $\sum_{t=0}^{\infty} t^{2}c^{t}$ converges in $x$ (and the other terms are positive constants), it follows that $\sum_{t=0}^{\infty} \partial_{u}\Delta_{v, t}(x)$ converges uniformly, as desired. 



\end{proof}

The proof of Theorem \ref{thm:impossibility} now follows immediately from Lemma \ref{lem:impossibility_uniform} and Theorem \ref{thm:exponential_implies_differentiable}.

\begin{proof}[Proof of Theorem \ref{thm:impossibility}]
By Lemma \ref{lem:impossibility_uniform}, there is no differentiable strong Bernoulli factory for the polytope $\P_{k,n}$. By Theorem \ref{thm:exponential_implies_differentiable}, any strong Bernoulli factory that converges exponentially is differentiable. It follows that there is no strong Bernoulli factory for $\P_{k,n}$ that converges exponentially.
\end{proof}

 \end{supplement}

\end{document}